%% file: main.tex
\pdfoutput=1

\documentclass[
  a4paper,
  USenglish,
  thm-restate,
  pdfa,
]{lipics-v2021}
\nolinenumbers
\input{preamble.tex}
\hideLIPIcs

\title{Optimal Offline ORAM with Perfect Security via Simple Oblivious Priority Queues}
\author{Thore Thie{\ss}en}{University of Münster, Germany}{t.thiessen@uni-muenster.de}{0009-0002-9902-3853}{}
\author{Jan Vahrenhold}{University of Münster, Germany}{jan.vahrenhold@uni-muenster.de}{0000-0001-8708-4814}{}
\authorrunning{T.\ Thie{\ss}en and J.\ Vahrenhold}
\Copyright{Thore Thie{\ss}en and Jan Vahrenhold}

\begin{document}
  \maketitle
  \relatedversiondetails[cite={conferenceVersion}]{Conference Version}{https://doi.org/10.4230/LIPIcs.ISAAC.2024.55}

  \begin{abstract}
    \emph{Oblivious RAM} (ORAM) is a well-researched primitive to hide the memory access pattern of a RAM computation;
    it has a variety of applications in trusted computing, outsourced storage, and multiparty computation.
    In this paper, we study the so-called \emph{offline} ORAM in which the sequence of memory access locations to be hidden is known in advance.
    Apart from their theoretical significance, offline ORAMs can be used to construct efficient oblivious algorithms.

    We obtain the first optimal offline ORAM with perfect security from oblivious priority queues via time-forward processing.
    For this, we present a simple construction of an oblivious priority queue with perfect security.
    Our construction achieves an asymptotically optimal (amortized) runtime of $\Th(\log N)$ per operation for a capacity of $N$ elements and is of independent interest.

    Building on our construction, we additionally present efficient external-memory instantiations of our oblivious, perfectly-secure construction:
    For the cache-aware setting, we match the optimal I/O complexity of $\Th(\frac{1}{B} \log \frac{N}{M})$ per operation (amortized), and for the cache-oblivious setting we achieve a near-optimal I/O complexity of $\Oh(\frac{1}{B} \log \frac{N}{M} \log\log_M N)$ per operation (amortized).
  \end{abstract}

  \input{section/introduction.tex}
  \input{section/priority-queue.tex}
  \input{section/offline-oram.tex}
  \input{section/external-priority-queue.tex}
  \input{section/conclusion.tex}

  \bibliography{literature}

  \newpage
  \appendix
  \input{section/k-selection.tex}
  \input{section/opq-operations.tex}
  \newpage
  \input{section/correct-purify-half.tex}
\end{document}

%% file: preamble.tex
\usepackage{csquotes}
\usepackage{mathtools}
\usepackage{xparse}
\usepackage{xspace}

\DeclareFontShape{T1}{lmr}{m}{scit}{<-> sub * lmr/m/scsl}{}
\sbox0{\sffamily X}
\DeclareFontShape{T1}{lmss}{bx}{sc}{<-> sub * cmr/bx/sc}{}
\DeclareFontShape{T1}{lmss}{bx}{scit}{<-> sub * cmr/bx/scsl}{}

\usepackage[ruled]{algorithm}
\usepackage[noEnd]{algpseudocodex}
\RenewDocumentCommand\Call{s m o m}{%
  \textproc{#2}%
  \IfValueT{#3}{\ifmmode_{#3}\else\textsubscript{#3}\fi}%
  \ifblank{#4}{\IfBooleanT{#1}{()}}{({#4})}}

\usepackage{booktabs}

\NewDocumentCommand\thead{m}{{\textit{#1}}}
\RenewDocumentCommand\tnote{s m}{{\IfBooleanT{#1}{\rlap}{\textsuperscript{#2}}}}

\DeclarePairedDelimiter\set{\lbrace}{\rbrace}
\DeclarePairedDelimiterX\setCond[2]{\lbrace}{\rbrace}{%
  {#1}\,\delimsize\vert\,\mathopen{}{#2}}
\DeclarePairedDelimiter\tuple{\langle}{\rangle}

\DeclarePairedDelimiter\paren{\lparen}{\rparen}
\DeclarePairedDelimiter\abs{\lvert}{\rvert}
\DeclarePairedDelimiter\floor{\lfloor}{\rfloor}
\DeclarePairedDelimiter\ceil{\lceil}{\rceil}

\NewDocumentCommand\Oh{}{\mathcal{O}}
\NewDocumentCommand\Om{}{\Omega}
\NewDocumentCommand\om{}{\omega}
\NewDocumentCommand\Th{}{\Theta}

\NewDocumentCommand\mult{}{\mathbin{\cdot}}
\NewDocumentCommand\concat{}{\mathbin{\|}}

\NewDocumentCommand\Tau{}{\mathcal{T}}

\usepackage{tikz}
\usetikzlibrary{
  arrows.meta,
  matrix,
  patterns,
  positioning,
}
\usepackage{pgfplots}
\pgfplotsset{compat=1.18}

\usepackage[capitalise,noabbrev]{cleveref}
\crefname{algocf}{Algorithm}{Algorithms}
\Crefname{algocf}{Algorithm}{Algorithms}
\crefname{inv}{invariant}{invariants}
\Crefname{inv}{Invariant}{Invariants}
\creflabelformat{inv}{#2(#1)#3}

\NewDocumentCommand\arb{o}{\,\cdot\,}
\NewDocumentCommand\addr{m}{{\operatorname{Addr}_{#1}}}
\NewDocumentCommand\rank{}{\operatorname{rank}}
\NewDocumentCommand\ea{}{~et~al.\xspace}
\NewDocumentCommand\eg{}{e.\,g.,\xspace}
\NewDocumentCommand\ie{}{i.\,e.,\xspace}

%% file: section/introduction.tex

\section{Introduction}%
\label{sec:Introduction}

Introduced by Goldreich and Ostrovsky~\cite{Goldreich.Ostrovsky96}, \emph{oblivious RAM} (\emph{ORAM}) conceals the memory access pattern of any RAM computation.
This prevents the leakage of confidential information when some adversary can observe the pattern of memory accesses.
We consider oblivious RAM in the offline setting:
This allows an additional pre-processing step on the access pattern while still requiring that the access pattern is hidden from the adversary.

Offline ORAMs can be used to construct efficient oblivious algorithms in situations where at least part of the memory access sequence is either known or can be inferred in advance.
As a motivating example, consider the classical Gale--Shapley algorithm for the stable matching problem~\cite{Gale.Shapley62,Mondal.Panda.ea23}:
In each round of the algorithm, up to $n$ parties make a proposal according to their individual preferences.
The preferences must be hidden to maintain obliviousness, and thus the memory access pattern may not depend on them.
While it seems that the standard algorithm makes online choices, in fact the preferences and the current matching are known before each round, so the proposals can be determined in advance and an offline ORAM can be used to hide the access pattern in each round.

Many of the previous works on offline (and online) ORAMs focus on statistical and computational security:
While optimal offline ORAMs are known for computational and statistical security~\cite{Asharov.Komargodski.ea22,Shi20}, the same is not true for perfect security.
We close this gap and obtain the first (asymptotically) optimal offline ORAM with perfect security.
We derive our construction from an oblivious priority queue.

For this, we discuss and analyze a construction of an oblivious priority queue simple enough to be considered part of folklore.
In fact, both the construction and its analysis can be used in an undergraduate data structures course as an example of how to construct an efficient oblivious data structure from simple building blocks.
Our construction reduces the problem to oblivious partitioning where an optimal oblivious algorithm~\cite{Asharov.Komargodski.ea22} is known.

\subsection{Oblivious Data Structures}%
\label{sec:ObliviousDataStructures}

Conceptually, (offline) ORAM and oblivious priority queues are \emph{oblivious data structures}.
Oblivious data structures provide efficient means to query and modify data while not leaking information, \eg distribution of the data or the operations performed, via the memory access pattern.
There are three main applications:
\begin{description}
  \item[Outsourced Storage] When storing data externally, oblivious data structures can be used in conjunction with encryption.
  Encryption alone protects the confidentiality of the data at rest, but performing operations may still leak information about queries or the data itself via the access pattern~\cite{Islam.Kuzu.ea12}.
  \item[Trusted Computing] When computing in trusted execution environments, oblivious data structures safeguard against many memory-related side channel attacks~\cite{Osvik.Shamir.ea06}.
  \item[{(Secure) Multiparty Computation}] In this setting, actors want to (jointly) compute a function without revealing their respective inputs to each other.
  Here, oblivious data structures have been used to allow for data structure operations with sublinear runtime~\cite{Toft11,Keller.Scholl14}.
\end{description}

\subsubsection{Security Definition}

In line with standard assumptions for oblivious algorithms~\cite{Goldreich.Ostrovsky96}, we assume the $w$-bit word RAM model of computation.
Let the random variable $\addr{\Call{Op}{x}}$ with
\begin{equation}
  \addr{\Call{Op}{x}} \in \paren*{\set{0, \ldots, 2^w - 1} \times \set{\Call{Read}{}, \Call{Write}{}}}^*
\end{equation}
denote the sequence of \emph{memory probes} for \Call{Op}{$x$}, \ie the sequence of memory access locations and memory operations performed by operation \Call{Op}{} for input $x$.
Access to a constant number of registers (\emph{private memory}) is excluded from the probe sequence.

For \emph{perfect} oblivious security, we require that all data structure operation sequences of length $n$ produce the same memory access pattern:

\begin{definition}[Obliviousness with Perfect Security]%
  \label{def:PerfectSecurity}
  We say that an (online) data structure $\mathcal{D}_N$ with capacity\footnote{%
    To hide the type of operation performed, in particular for intermixed \Call{Insert}{} and \Call{Delete}{} sequences, it is assumed that the data structure has a fixed capacity $N$ determined a priori.
    This assumption does not limit any of our analyses, as the capacity can be adjusted using standard (doubling) techniques with (amortized) constant asymptotic overhead per operation.
  } $N$ and operations $\Call{Op}[1]{}, \ldots, \Call{Op}[m]{}$ is oblivious with perfect security iff, for every two sequences of $n$ operations
  \begin{equation*}
    X = \tuple{\Call{Op}[i_1]{x_1}, \ldots, \Call{Op}[i_n]{x_n}}
    \quad\text{and}\quad
    Y = \tuple{\Call{Op}[j_1]{y_1}, \ldots, \Call{Op}[j_n]{y_n}}
  \end{equation*}
  with valid inputs $x_k$, $y_k$, the memory probe sequences are identically distributed, \ie
  \begin{equation*}
    \tuple{\addr{\Call{Op}[i_1]{x_1}}, \ldots, \addr{\Call{Op}[i_n]{x_n}}}
    \equiv
    \tuple{\addr{\Call{Op}[j_1]{y_1}}, \ldots, \addr{\Call{Op}[j_n]{y_n}}}
    \,\text{.}
  \end{equation*}
\end{definition}

The requirement of identical distribution in the above definition can be relaxed to strictly weaker definitions of security by either allowing a negligible statistical distance of the probe sequences (\emph{statistical security}) or allowing a negligible distinguishing probability by a polynomial-time adversary (\emph{computational security});
see Asharov\ea~\cite{Asharov.Komargodski.ea22} for more details.

\Cref{def:PerfectSecurity} immediately implies that the memory probe sequence is independent of the operation arguments --- and, by extension, the data structure contents --- as well as the operations performed (\emph{operation-hiding security}).
As a technical remark, we note that for perfectly-secure data structure operations with determined outputs, the joint distributions of output and memory probe sequence are also identically distributed.
This implies that data structures satisfying \cref{def:PerfectSecurity} are universally composable~\cite{Asharov.Komargodski.ea22}.

\subsubsection{Offline ORAM}

The (online) ORAM is essentially an oblivious array data structure~\cite{Larsen.Nielsen18}.
By using an ORAM as the main memory, any RAM program can generically be transformed into an oblivious program at the cost of an \emph{overhead} per memory access.

The offline ORAM we are considering here, however, is given the sequence $I$ of access locations in advance.
While this allows pre-computations on $I$, the probe sequence must still hide the operations and indices in $I$.
In anticipation of the offline ORAM construction in \cref{sec:OfflineOram}, we take a similar approach as Mitchell and Zimmerman~\cite{Mitchell.Zimmerman14} and define an offline ORAM as an \emph{online} oblivious data structure with additional information:

\begin{definition}[Offline ORAM]%
  \label{def:OfflineORAM}
  An offline ORAM is an oblivious data structure $\mathcal{D}_N$ that maintains an array of length $N$ under an annotated online sequence of read and write operations:
  \begin{description}
    \item[\Call{Read}{$i, \tau$}] Return the value stored at index $i$ in the array.
    \item[\Call{Write}{$i, v', \tau$}] Store the value $v'$ in the array at index $i$.
  \end{description}
  The annotation $\tau$ indicates the time-stamp of the next operation accessing index $i$.
\end{definition}

Note that this definition implies that $\mathcal{D}_N$ can also be used in an online manner if the time-stamps $\tau$ of the next operation accessing the index $i$ are known.
When discussing the offline ORAM construction in \cref{sec:OfflineOram}, we show how to use sorting and linear scans to compute the annotations $\tau$ from the sequence $I$ of access locations given in advance.

\begin{table}[t]\centering
  \begin{threeparttable}
    \begin{tabular*}{\textwidth}{@{\extracolsep\fill} c c c c l}
      \toprule
      \thead{security} & \thead{runtime} & \thead{priv.\ memory} & \thead{deletion} & \\
      \midrule
      perfect\tnote*{p} & $\Oh(\log^2 N)$\tnote*{a} & $\Oh(1)$ & no & \cite{Toft11} \\
      statistical & $\Oh(\log^2 N)$ & $\Oh(\om(1) \mult \log N)$ & no & \cite{Wang.Nayak.ea14} \\
      statistical & $\Oh(\log^2 N)$ & $\Oh(\om(1) \mult \log N)$ & yes\tnote*{r} & \cite[Path ORAM variant]{Keller.Scholl14} \\
      perfect & $\Oh(\log^2 N)$\tnote*{a} & $\Oh(1)$ & no & \cite{Mitchell.Zimmerman14} \\
      statistical & $\Oh(\log N)\tnote*{a}$ & $\Oh(\om(1) \mult \log N)$ & yes & \cite{Jafargholi.Larsen.ea21} \\
      statistical & $\Oh(\om(1) \mult \log N)$ & $\Oh(1)$ & yes\tnote*{r} & \cite[Circuit variant]{Shi20} \\
      perfect & $\Oh(\log^2 N)$\tnote*{a} & $\Oh(1)$ & yes\tnote*{r} & \cite{Ichikawa.Ogata23} \\
      \midrule
      perfect & $\Oh(\log N)$\tnote*{a} & $\Oh(1)$ & no & \textbf{new} \\
      \bottomrule
    \end{tabular*}
    \begin{tablenotes}[para]
      \item[p] reveals the operation
      \item[a] amortized runtime complexity
      \item[r] requires an additional reference
    \end{tablenotes}
    \caption{%
      Oblivious priority queues supporting \Call{Insert}{}, \Call{Min}{}, and \Call{DeleteMin}{}.
      Deletions are noted as supported if an operation \Call{Delete}{}, \Call{ModifyPriority}{}, or \Call{DecreasePriority}{} is available.}%
  \label{tbl:PriorityQueueConstructions}
  \end{threeparttable}
\end{table}

\subsection{Previous Work}%
\label{sec:PreviousWork}

\subparagraph{Oblivious Priority Queues}

Because of their many algorithmic applications, oblivious priority queues have been considered in a number of previous works.
We provide an overview of previous oblivious priority queue constructions in \cref{tbl:PriorityQueueConstructions}.

Jacob\ea~\cite{Jacob.Larsen.ea19} show that a runtime of $\Om(\log N)$ per operation is necessary for oblivious priority queues.
Their lower bound holds even when allowing a constant failure probability and relaxing the obliviousness to statistical or computational security.

The first oblivious priority queue construction due to Toft~\cite{Toft11} is perfectly-secure and has an amortized runtime of $\Oh(\log^2 N)$, but reveals the operation performed and lacks operations to delete or modify arbitrary elements.
Subsequent perfectly-secure constructions~\cite{Mitchell.Zimmerman14,Ichikawa.Ogata23} offer operation-hiding security or support additional operations, but do not improve the suboptimal $\Oh(\log^2 N)$ runtime.
A different line of work considers oblivious priority queues with statistical security.
Jafargholi\ea~\cite{Jafargholi.Larsen.ea21} and, subsequently, Shi~\cite{Shi20} both present constructions with an optimal $\Th(\log N)$ runtime.
All statistically secure priority queue constructions~\cite{Wang.Nayak.ea14,Keller.Scholl14,Jafargholi.Larsen.ea21,Shi20} are randomized;
many~\cite{Wang.Nayak.ea14,Keller.Scholl14,Shi20} also rely on tree-based ORAMs (\eg \emph{Path ORAM}~\cite{Stefanov.Dijk.ea13} or \emph{Circuit ORAM}~\cite{Chan.Shi17}) in a non--black-box manner.

\subparagraph{Offline ORAMs}

\begin{table}[t]\centering%
  \begin{threeparttable}
    \begin{tabular*}{\textwidth}{@{\extracolsep\fill} c c c c c c c}
      \toprule
      & \multicolumn{2}{c}{\thead{perfect security}} & \multicolumn{2}{c}{\thead{statistical security}} & \multicolumn{2}{c}{\thead{comput.\ security}} \\
      \midrule
      \multirow{2}{*}{\thead{online}} & $\Om(\log N)$ & \cite{Larsen.Nielsen18} & $\Om(\log N)$ & \cite{Larsen.Nielsen18} & $\Om(\log N)$ & \cite{Larsen.Nielsen18} \\
      & $\Oh(\log^3 N / \log\log N)$ & \cite{Chan.Shi.ea21} & $\Oh(\log^2 N)$ & \cite{Chan.Shi17} & $\Oh(\log N)$\tnote*{p} & \cite{Asharov.Komargodski.ea23} \\
      \addlinespace[.5ex]%
      \cmidrule{2-7}%
      \addlinespace[.5ex]%
      \multirow{3}{*}{\thead{offline}} & $\Om(\log N)$\tnote*{i} & \cite{Goldreich.Ostrovsky96} & $\Om(\log N)$\tnote*{i} & \cite{Goldreich.Ostrovsky96,Boyle.Naor16} & $\Om(1)$ & trivial~\cite{Boyle.Naor16} \\
      & $\Oh(\log^2 N)$\tnote*{a} & \eg via~\cite{Mitchell.Zimmerman14} & $\Oh(\om(1) \mult \log N)$ & \cite{Shi20} & $\Oh(\log N)$\tnote*{p} & \cite{Asharov.Komargodski.ea23} \\
      & $\Oh(\log N)$\tnote*{a} & \textbf{new} & & & & \\
      \bottomrule
    \end{tabular*}
    \begin{tablenotes}[para]
      \item[p] assuming a pseudo-random function family
      \item[i] assuming indivisibility~\cite{Boyle.Naor16}
      \item[a] amortized
    \end{tablenotes}
    \caption{%
      Best known overhead bounds for online and offline ORAMs with $N$ memory cells, a constant number of private memory cells, and standard parameters~\cite{Larsen.Nielsen18}.}%
  \label{tbl:OfflineORAMConstructions}
  \end{threeparttable}
\end{table}

Though much of the research focuses on online ORAMs, \emph{offline} ORAMs have been explicitly considered in some previous works~\cite{Mitchell.Zimmerman14,Boyle.Naor16,Jafargholi.Larsen.ea21,Shi20}.
We provide an overview of the best known upper and lower bounds for both online and offline ORAM constructions with perfect, statistical, or computational security in \cref{tbl:OfflineORAMConstructions}.

Goldreich and Ostrovsky~\cite{Goldreich.Ostrovsky96} prove a lower bound on the overhead of $\Om(\log N)$ for (online) ORAMs with perfect security (assuming indivisibility).
This bound also applies to offline ORAMs and constructions with statistical security~\cite{Boyle.Naor16}.

There is a generic way to construct offline ORAMs from oblivious priority queues (see \cref{sec:OfflineOram}).
Via their priority queue construction, Shi~\cite{Shi20} obtains an optimal offline ORAM with statistical security for a private memory of constant size.
For computational security, the state-of-the-art online ORAM construction~\cite{Asharov.Komargodski.ea23} is simultaneously the best known offline construction (asymptotically).
While the upper bounds for statistical and computational security match the (conjectured) $\Om(\log N)$ lower bound, prior to our work there remained a gap for perfect security.\footnote{%
  Boyle and Naor~\cite{Boyle.Naor16} show how to construct an offline ORAM with overhead $\Oh(\log N)$:
  In addition to the access locations, their construction must be given the \emph{operands} of the write operations in advance, \ie the sequence of values to be written.
  It thus does not fit our more restrictive \cref{def:OfflineORAM}.}

\subsection{Contributions}

Our work provides several contributions to a better understanding of the upper bounds of perfectly-secure oblivious data structures:
\begin{itemize}
  \item As a main contribution, we present and analyze an oblivious priority queue construction with perfect security.
  This construction is conceptually simple and achieves the optimal $\Th(\log N)$ runtime per operation amortized.

  In particular, our construction improves over the previous statistically-secure constructions~\cite{Jafargholi.Larsen.ea21,Shi20} in that we eliminate the failure probability (perfect security with perfect correctness) and achieve a strictly-logarithmic runtime for $\Oh(1)$ private memory cells.\footnote{%
    For a private memory of constant size, the construction of Shi~\cite{Shi20} requires an additional $\om(1)$-factor in runtime to achieve a negligible failure probability.}

  \item The priority queue implies an optimal $\Th(\log N)$-overhead offline ORAM with perfect security, closing the gap to statistical and computational security in the offline setting.
  We show that these bounds hold even for a large number $n$ of operations, \ie $n = N^{\om(1)}$.

  \item We also provide improved external-memory oblivious priority queues:
  Compared to the I/O-optimal state-of-the-art~\cite{Jafargholi.Larsen.ea21}, our cache-aware construction achieves perfect security and only requires a private memory of constant size.

  In the cache-oblivious setting, our construction achieves near-optimal I/O-complexity for perfect security and a private memory of constant size.
  We are not aware of any previous oblivious priority queues in the cache-oblivious setting.
\end{itemize}

%% file: section/priority-queue.tex

\section{Oblivious Priority Queue from Oblivious Partitioning}%
\label{sec:ObliviousPriorityQueue}

An oblivious priority queue data structure maintains up to $N$ elements and must support at least three non-trivial operations prescribed by the abstract data type \textsf{PriorityQueue}:
\begin{description}
  \item[\Call*{Insert}{$k, p$}] Insert the element $\tuple{k, p}$ with \emph{priority} $p$.
  \item[\Call*{Min}{}] Return the element $\tuple{k, p_{\text{min}}}$ with the minimal priority $p_{\text{min}}$.
  \item[\Call*{DeleteMin}{}] Remove the element $\tuple{k, p_{\text{min}}}$ with minimal priority $p_{\text{min}}$.
\end{description}
We assume that both the key $k$ and the priority $p$ fit in a constant number of memory cells and that the relative order of two priorities $p, p'$ can be determined obliviously in constant time;
larger elements introduce an overhead factor in the runtime.
To keep the exposition simple, we assume distinct priorities.
This assumption can be removed easily, see \cref{sec:NonDistinctPriorities}.

\Cref{fig:PriorityQueueStructure} shows the structure of our solution:
In a standard data structure layout, it has $\ell \in \Th(\log N)$ levels of geometrically increasing size.
Each level $i$ consists of a \emph{down}-buffer $D_i$ and an \emph{up}-buffer $U_i$, both of size $\Th(2^i)$.
\Call{Insert}{} inserts into the up-buffer $U_0$ and \Call{DeleteMin}{} removes from the down-buffer $D_0$.
Each level $i$ is rebuilt after $2^i$ operations, moving elements up through the up-buffers and back down through the down-buffers.
The main idea guiding the rebuilding is to ensure that all levels $j < i$ can support the next operations until level $i$ is rebuilt;
we later formalize this as an invariant for the priority queue (see \cref{lma:Invariants}).

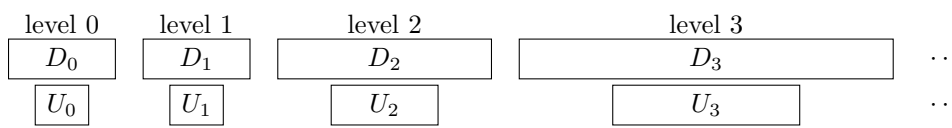
\begin{figure}[t]\centering
  \input{figure/priority-queue.tikz}%
  \caption{%
    Structure of the oblivious priority queue:
    Each level $i \in \set{0, \ldots, \ell - 1}$ consists of a down-buffer $D_i$ and an up-buffer $U_i$ half the size of $D_i$.}%
  \label{fig:PriorityQueueStructure}
\end{figure}

\subparagraph{Oblivious Building Blocks}%
\label{sec:ObliviousBuildingBlocks}

For our construction, we need an algorithm to obliviously permute a given array $A$ of $n$ elements such that the $k$ smallest elements are swapped to the front, followed by the remaining $n - k$ elements.
We refer to this problem as \emph{$k$-selection}.

To obtain an efficient algorithm for $k$-selection, we use an oblivious modification of the classical (RAM) linear-time selection algorithm~\cite{Blum.Floyd.ea73} as sketched by Lin\ea~\cite[full version, Appendix~E.2]{Lin.Shi.ea19}.
This reduces $k$-selection to the \emph{partitioning} problem (also called \emph{$1$-bit sorting}).
The oblivious $k$-selection deviates from the classical algorithm in two respects~\cite{Lin.Shi.ea19}:
\begin{itemize}
  \item First, it is necessary to ensure that the partitioning step is oblivious.
  For this, instead of the algorithms proposed by Lin\ea, we use the optimal oblivious partitioning algorithm of Asharov\ea~\cite[Theorem~5.1]{Asharov.Komargodski.ea22}.\footnote{%
    Note that Asharov\ea\ refer to the partitioning problem as \emph{compaction}.
    We use the term partitioning to stress that all elements of the input are preserved which is necessary for our definition of $k$-selection.}
  This allows us to obtain a linear-time algorithm for $k$-selection.

  \item Second, the relative position of the median of medians among the elements cannot be revealed as this would leak information about the input.
  To address this, Lin\ea\ propose over-approximating the number of elements and always recursing with approximately $\frac{7 n}{10}$ elements.
\end{itemize}

\begin{corollary}[Oblivious $k$-Selection via~\cite{Lin.Shi.ea19,Asharov.Komargodski.ea22}]\label{cor:ObliviousKSelection}
  There is a deterministic, perfectly-secure oblivious algorithm for the $k$-selection problem with runtime $\Oh(n)$ for $n$ elements.
\end{corollary}

We describe the algorithm in detail and prove its correctness in \cref{sec:ObliviousKSelection}.

\subparagraph{Comparison with Jafargholi\ea~\cite{Jafargholi.Larsen.ea21}}

Conceptually, our construction is similar to that of Jafargholi\ea~\cite{Jafargholi.Larsen.ea21}:
In both constructions, the priority queue consists of levels of geometrically increasing size with lower-priority elements moving towards the smaller levels.
Structuring the construction so that larger levels are rebuilt less frequently is a standard data structure technique to amortize the cost of rebuilding.

The main difference lies in the rebuilding itself:
In the construction of Jafargholi\ea, level $i$ is split into $2^i$ nodes;
overall, the levels form a binary tree.
The elements are then assigned to paths in the tree based on their key~\cite{Jafargholi.Larsen.ea21}.
While this allows deleting elements by their key efficiently, this inherently introduces the probability of \enquote{overloading} certain nodes, reducing the construction to statistical security (with a negligible failure probability).

We instead use $k$-selection for rebuilding;
this allows us to maintain both perfect correctness and security.
Unfortunately, this comes at the cost of a more expensive \Call{Delete}{} operation:
Since we maintain no order on the keys within each level, it is not possible to efficiently delete arbitrary elements by their key.
We note that deleting arbitrary elements is not required for our offline ORAM construction.

\subsection{Details of the Construction}%
\label{sec:ObliviousPriorityQueueDetails}

The priority queue consists of $\ell \coloneqq \ceil{\log_2 N}$ levels, each with a down-buffer $D_i$ of $2^{\max\set{1, i}}$ elements and an up-buffer $U_i$ of $2^{\max\set{0, i - 1}} = \frac{\abs{D_i}}{2}$ elements.
An element is a pair $\tuple{k, p}$ of key $k$ and priority $p$;
each buffer is padded with \emph{dummy} elements to hide the number of \enquote{real} elements.
Initially, all elements in the priority queue are dummy elements.
We refer to a buffer containing only dummy elements as \emph{empty}.

The elements are distributed over the levels via a rebuilding procedure:
Level $i$ is rebuilt after exactly $2^i$ operations.
Let $\Delta_i$ be the remaining number of operations until level $i$ is rebuilt (with $\Delta_i = 2^i$ initially).
After each operation, all counters $\Delta_i$ are decremented by one and all levels $i$ with $\Delta_i = 0$ are rebuilt with \Call{Rebuild}{$m$} for $m \coloneqq \max\setCond{i < \ell}{\Delta_i = 0}$;
note that $\Delta_i = 0$ if and only if $i \leq m$.
The counter $\Delta_i$ of each rebuilt level $i \leq m$ is reset to $2^i$, so $\Delta_i > 0$ for every level $i$ after each operation.

We will show the correctness of the construction with three \cref{lma:Invariants:CorrectElements,lma:Invariants:U0Empty,lma:Invariants:ElementRanks}:

\begin{lemma}[Invariants]\label{lma:Invariants}
  Before each operation of the priority queue, the following holds:
    \begin{enumerate}[(a)]
      \item The priority queue contains the correct elements, \ie $\mathcal{E} = \bigcup_{i < \ell} \paren{U_i \cup D_i}$ where $\mathcal{E}$ denotes the elements that should be contained in the priority queue (with standard semantics).%
        \label[inv]{lma:Invariants:CorrectElements}
      \item The up-buffer $U_0$ is empty, \ie contains exactly one dummy element.%
        \label[inv]{lma:Invariants:U0Empty}
      \item For all elements $e \coloneqq \tuple{k, p} \in D_i \cup U_i$ with $i \geq 1$, it holds that $\Delta_i \leq \rank(e)$ where $\rank(\tuple{\arb[k], p}) \coloneqq \abs*{\setCond{\tuple{\arb[k'], p'} \in \mathcal{E}}{p' < p}}$ is the (unique) rank of $p$ in the priority queue.%
        \label[inv]{lma:Invariants:ElementRanks}
    \end{enumerate}
\end{lemma}

The most important \cref{lma:Invariants:ElementRanks} guarantees that each level $i \geq 1$ is rebuilt before any of its elements are required for \Call{Min}{}/\Call{DeleteMin}{} in $D_0$.
In turn, this implies that the $\Delta_i$ smallest elements potentially required before rebuilding level $i$ are stored in the buffers on levels $0, \ldots, i - 1$.
Formally, this follows since $\Delta_j \leq \Delta_i$ for all $j < i$.

\begin{figure}[t]\centering
  \input{figure/priority-queue-rebuild.tikz}%
  \caption{%
    Distribution of the elements when rebuilding level $m$:
    The up to $2^{m + 1}$ smallest elements in the levels $0, \ldots, m$ are distributed over the down-buffers of the first $m$ levels.
    The up to $2^m$ remaining elements are inserted into the (empty) up-buffer $U_{m + 1}$ of level $m + 1$.}%
  \label{fig:PriorityQueueRebuild}
\end{figure}
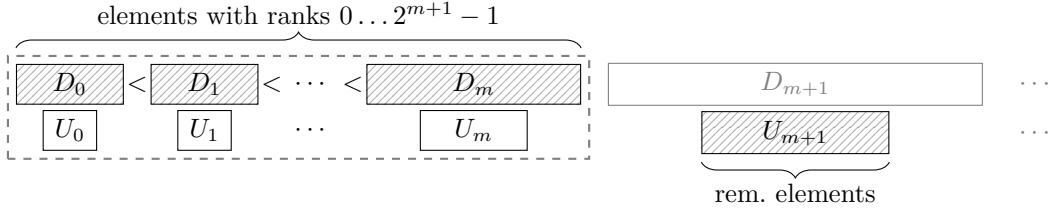

For simplicity of exposition, we ignore the details of the rebuilding step for the time being and discuss how the three priority queue operations can be implemented while maintaining the \cref{lma:Invariants:CorrectElements,lma:Invariants:ElementRanks} of \cref{lma:Invariants}:

\begin{description}
  \item[\Call{Insert}{$k, p$}]
    The dummy in $U_0$ is replaced with the new element $\tuple{k, p}$.
    After the operation, the priority queue contains the elements $\mathcal{E}' = \mathcal{E} \cup \set{\tuple{k, p}} = \bigcup_{i < \ell} \paren{U_i \cup D_i}$.
    Inserting a new element does not decrease the rank of any element while all counters $\Delta_i$ decrease;
    this implies that \cref{lma:Invariants:ElementRanks} is maintained.

  \item[\Call*{Min}{}]
    The minimal element $e_{\text{min}}$ with $\rank(e_{\text{min}}) = 0$ must be contained in level $0$ since $\Delta_i > 0$ for all $i > 0$ before each operation.
    Since $U_0$ is empty, $e_{\text{min}}$ is one of the two elements in $D_0$.
    After the operation, the elements $\mathcal{E}' = \mathcal{E} = \bigcup_{i < \ell} \paren{U_i \cup D_i}$ remain the same.
    \Cref{lma:Invariants:ElementRanks} is maintained since the ranks of all elements $e \in \mathcal{E}$ remain unchanged.

  \item[\Call*{DeleteMin}{}]
    For this operation, we replace the minimal element $e_{\text{min}} \in D_0$ with a dummy element.
    After the operation, the priority queue contains the elements $\mathcal{E}' = \mathcal{E} \setminus \set{e_{\text{min}}} = \bigcup_{i < \ell} \paren{U_i \cup D_i}$.
    Removing the minimum reduces the rank of all other elements by one, but \cref{lma:Invariants:ElementRanks} is maintained since all counters $\Delta_i$ also decrease.
\end{description}

For the operation-hiding security, we access memory locations for all three operations but only perform updates for the intended operation.
For example, \Call{DeleteMin}{} will access both $U_0$ and $D_0$, but only actually overwrite the minimal element in $D_0$ with a dummy element.
We provide pseudocode for the operations in \cref{sec:ObliviousPriorityQueueOperations}.

We now turn to describing \Call{Rebuild}{$m$} (\cref{alg:Rebuild}).
As shown in \cref{fig:PriorityQueueRebuild}, this procedure processes all elements in the levels $0, \ldots, m$:
The non-dummy elements are distributed into $D_0, \ldots, D_m, U_{m + 1}$ and the up-buffers $U_0, \ldots, U_m$ are emptied, \ie filled with dummy elements.
The down-buffers $D_0, \ldots, D_m$ collectively contain up to $2^{m + 1}$ non-dummy elements.
Additionally, the up-buffers $U_0, \ldots, U_m$ collectively contain up to $2^m$ non-dummy elements.
All these elements are distributed over the buffers $D_0, \ldots, D_m$ and $U_{m + 1}$ such that
\begin{itemize}
  \item $D_0$ contains the two smallest elements (with ranks $0$ and $1$),
  \item the other $D_i$ (for $i \leq m$) each contain the elements with ranks $2^i, \ldots, 2^{i + 1} - 1$,\footnote{%
    Even if there are more than $2^{m + 1}$ elements in the priority queue overall, the buffer $D_m$ may still end up (partially) empty after rebuilding.
    This is no threat to the correctness, since \cref{lma:Invariants:ElementRanks} guarantees that level $m + 1$ will be rebuilt before requiring the elements with ranks $\geq 2^m$.%
  } and
  \item $U_{m + 1}$ contains all remaining elements.
\end{itemize}
For this, we order the elements by their priority $p$;
dummy elements have no priority and are ordered after non-dummy elements.

\begin{algorithm}[t]
  \caption{%
    Rebuild the levels $0, \ldots, m$ in the oblivious priority queue.
    Let $A \concat B$ denote the concatenation of two buffers $A$ and $B$;
    $A_{0 \ldots i}$ denotes the concatenation $A_0 \concat \cdots \concat A_i$.}%
  \label{alg:Rebuild}
  \begin{algorithmic}[1]
    \Procedure{Rebuild}{$m$}
      \State\Call{KSelect}{$2^{m + 1}, D_{0 \ldots m} \concat U_{0 \ldots m}$}%
        \label{alg:Rebuild:FirstSelect}%
        \Comment{move $2^{m + 1}$ smallest elements to $D_{0 \ldots m}$}
      \If{$m$ is not the last level}
        \State $U_{m + 1} \gets U_{0 \ldots m}$%
          \label{alg:Rebuild:Copy}%
          \Comment{copy the elements in $U_{0 \ldots m}$ to $U_{m + 1}$}
        \State $U_{0 \ldots m} \gets \tuple{\bot, \ldots, \bot}$%
          \Comment{overwrite $U_{0 \ldots m}$ with dummy elements}
      \EndIf
      \For{$i \gets m - 1, \ldots, 0$}
        \State\Call{KSelect}{$2^{i + 1}, D_{0 \ldots i + 1}$}
          \Comment{move $2^{i + 1}$ smallest elements to $D_{0 \ldots i}$}
      \EndFor
      \For{$i \gets 0, \ldots, m$}
        \State $\Delta_i \gets 2^i$
          \Comment{reset counters}
      \EndFor
    \EndProcedure
  \end{algorithmic}
\end{algorithm}

We can now prove that the overall construction is correct by showing that \Call{Rebuild}{$m$} with $m \coloneqq \max\setCond{i < \ell}{\Delta_i = 0}$ maintains the \cref{lma:Invariants:CorrectElements,lma:Invariants:U0Empty,lma:Invariants:ElementRanks}:

\begin{proof}[Proof of \cref{lma:Invariants}]
  All \namecrefs{lma:Invariants:CorrectElements} trivially hold for the empty priority queue.
  As described above, the operations \Call{Insert}{}, \Call{Min}{}, and \Call{DeleteMin}{} maintain \cref{lma:Invariants:CorrectElements,lma:Invariants:ElementRanks}.
  After each operation, all counters $\Delta_i$ are decremented and the levels $i$ with $\Delta_i = 0$ are rebuilt.
  We now show that all \namecrefs{lma:Invariants:CorrectElements} hold after rebuilding.

  \proofsubparagraph{\Cref{lma:Invariants:CorrectElements,lma:Invariants:U0Empty}: \boldmath$\mathcal{E} = \bigcup_{i < \ell} \paren{U_i \cup D_i}$ and \boldmath$U_0$ is empty}
  We first show that prior to \Call{Rebuild}{$m$} in each operation where $m$ is not the last level, the up-buffer $U_{m + 1}$ is empty.
  This can be seen by considering the two possible cases:
  \begin{itemize}
    \item If no more than $2^m$ operation have been performed overall, the level $m + 1$ has never been accessed.
      In this case $U_{m + 1}$ is empty since it was empty initially.

    \item Otherwise, if more than $2^m$ operations have been performed, the up-buffer $U_{m + 1}$ was emptied $2^m$ operations before by \Call{Rebuild}{$m'$} for some $m' > m$ (and not accessed since).
  \end{itemize}
  This means that by copying the elements in $U_{0 \ldots m}$ into $U_{m + 1}$ (\cref{alg:Rebuild:Copy}), only dummy elements are being overwritten and \cref{lma:Invariants:CorrectElements} is maintained.

  In case $m$ is the last level ($m = \ell - 1$), after \cref{alg:Rebuild:FirstSelect} the up-buffers $U_0, \ldots, U_m$ are empty iff there no more than $2^{m + 1} = 2^\ell \geq N$ elements in the data structure.
  This is guaranteed by the capacity bound $N$.
  Thus, the up-buffer $U_0$ is empty after each operation.

  \proofsubparagraph{\Cref{lma:Invariants:ElementRanks}: \boldmath$\Delta_i \leq \rank(e)$ for all $e \in U_i \cup D_i$ with $i \geq 1$}
  Next, we show that rebuilding maintains the rank invariant for all redistributed elements.
  Using $k$-selections to redistribute the elements makes sure that a buffer in level $i$ receives non-dummy elements only if all $D_0, \ldots, D_{i - 1}$ have been filled to capacity.
  Consider any level $i \geq 1$:
  If an element $e \coloneqq \tuple{k, p}$ is redistributed into level $i$, exactly $2^i = \sum_{j < i} \abs{D_j}$ elements $\tuple{\arb[k'], p'}$ with $p' < p$ must have been redistributed into lower levels, so $2^i \leq \rank(e)$.
  Thus, for all non-dummy elements $e$ inserted into a level $i \geq 1$, it holds that $\Delta_i \leq 2^i \leq \rank(e)$.
  For elements that remain in a level $i > m$, \cref{lma:Invariants:ElementRanks} is trivially maintained.
\end{proof}

With this, we obtain our perfectly-secure priority queue construction:

\begin{theorem}[Optimal Oblivious Priority Queue]\label{thm:ObliviousPriorityQueue}
  There is a deterministic, perfectly-secure oblivious priority queue with capacity $N$ that supports each operation in amortized $\Oh(\log N)$ time and uses $\Oh(N)$ space.
\end{theorem}

\begin{proof}
  Apart from the rebuilding, the runtime for \Call{Insert}{}, \Call{Min}{}, and \Call{DeleteMin}{} is constant.
  The amortized runtime per operation for \Call{Rebuild}{} is bounded by
  \begin{align*}
    & \sum_{m = 0}^{\ell - 1} \frac{T_{\Call{KSelect}{}}(2^{m + 1} + 2^m) + \sum_{i = 0}^{m - 1} T_{\Call{KSelect}{}}(2^{i + 2}) + c \mult 2^m}{2^m} \\
    \leq{} & \sum_{m = 0}^{\ell - 1} \frac{\Oh(2^m)}{2^m}
    \in \Oh(\ell)
    = \Oh(\log N)
    \,\text{.}
  \end{align*}
  The space bound follows immediately since all algorithmic building blocks have a linear runtime and the combined size of all up- and down-buffers is linear in $N$.

  By using the deterministic, perfectly-secure algorithm for $k$-selection (\cref{cor:ObliviousKSelection}), the obliviousness follows since the access pattern for each operation is a deterministic function of the capacity $N$ and the number of operations performed so far.
\end{proof}

Due to the lower bound for oblivious priority queues~\cite{Jacob.Larsen.ea19}, this runtime is optimal.
If the type of operation does not need to be hidden, \Call{Min}{} can be performed in constant time since rebuilding is only required for correctness when adding or removing an element.
By applying \Call{Rebuild}{$\ell - 1$} directly, the priority queue can be initialized from up to $N$ elements in $\Oh(N)$ time.

\subsection{Non-Distinct Priorities}%
\label{sec:NonDistinctPriorities}

For non-distinct priorities, we want to ensure that ties are broken such that the order of insertion is preserved, \ie that elements inserted earlier are extracted first.
For this, we augment each element with the time-stamp $t$ of the \Call{Insert}{} operation and order the elements lexicographically by priority and time-stamp.
This only increases the size of each element by a constant number of memory cells and thus does not affect the runtime complexity.

It remains to bound the size of the time-stamp for a super-polynomial number of operations:\footnote{%
  Shi~\cite[Section~III.E]{Shi20} also address this issue, but for our amortized construction we can use a simpler approach based on oblivious sorting.}
Here we note that when rebuilding the last level (\Call{Rebuild}{$\ell - 1$}), we can additionally sort all elements by time-stamp and compress the time-stamps to the range $\set{0, \ldots, N - 1}$ (preserving their order).
We then assign time-stamps starting with $t = N$ until the last level is rebuilt again.
This ensures that $2 N$ is an upper bound for the time-stamps, so $\Oh(\log N)$ bits suffice for each time-stamp.
With an (optimal) oblivious $\Oh(n \log n)$-time sorting algorithm~\cite{Ajtai.Komlós.ea83}, the additional amortized runtime for sorting is bounded by
\begin{equation}
  \frac{1}{2^{\ell - 1}} \mult T_{\Call{Sort}{}}(N)
  \in \Th(\log N)
\end{equation}
and does not affect the overall runtime complexity.

%% file: figure/priority-queue.tikz

\begin{tikzpicture}[
  node distance=.2em and 1em,
  block/.style={
    draw,rectangle,
    minimum height=1.5em,
    minimum width=#1,
  },
  caption/.style={
    inner sep=0,
  },
]

  \node[block=4em] (D0) {$D_0$};
  \node[block=2em,below=of D0] (U0) {$U_0$};
  \node[caption,above=of D0] {level $0$};

  \node[block=4em,right=of D0] (D1) {$D_1$};
  \node[block=2em,below=of D1] (U1) {$U_1$};
  \node[caption,above=of D1] {level $1$};

  \node[block=8em,right=of D1] (D2) {$D_2$};
  \node[block=4em,below=of D2] (U2) {$U_2$};
  \node[caption,above=of D2] {level $2$};

  \node[block=14em,right=of D2] (D3) {$D_3$};
  \node[block=7em,below=of D3] (U3) {$U_3$};
  \node[caption,above=of D3] {level $3$};

  \node[right=of D3] (Di) {$\cdots$};
  \node[right=4.5em of U3] {$\cdots$};
\end{tikzpicture}

%% file: figure/priority-queue-rebuild.tikz

\begin{tikzpicture}[
  node distance=.2em and 1em,
  block/.style={
    draw,rectangle,
    minimum height=1.5em,
    minimum width=#1,
  },
  usedBlock/.style={
    block=#1,
    pattern=north east lines,pattern color=black!30,
  },
  unusedBlock/.style={
    block=#1,
    gray,
  },
  halfUsedBlock/.style={
    block=#1,
    path picture={\fill[black!10] (path picture bounding box.south west) rectangle (path picture bounding box.north);},
  },
  caption/.style={
    inner sep=0,
    align=center,
  },
  brace/.style={
    decorate,
    decorate,decoration={brace,amplitude=6pt},
  },
  braceCaptionAbove/.style={
    caption,
    midway,
    above,yshift=8pt,
  },
  braceCaptionBelow/.style={
    caption,
    midway,
    below,yshift=-8pt,
  },
]

  \node[usedBlock=4em] (D0) {$D_0$};
  \node[block=2em,below=of D0] (U0) {$U_0$};

  \node[usedBlock=4em,right=of D0] (D1) {$D_1$};
  \node[block=2em,below=of D1] (U1) {$U_1$};
  \path (D0) -- (D1) node[midway] {$<$};

  \node[right=of D1] (D2) {$\cdots$};
  \node[right=2em of U1] (U2) {$\cdots$};
  \path (D1) -- (D2) node[midway] {$<$};

  \node[usedBlock=8em,right=of D2] (D3) {$D_m$};
  \node[block=4em,below=of D3] (U3) {$U_m$};
  \path (D2) -- (D3) node[midway] {$<$};

  \draw[thick,dashed,gray] ($(U0.south west) + (-1.3em,-.3em)$) rectangle ($(D3.north east) + (.3em,.3em)$);

  \node[unusedBlock=14em,right=of D3] (D4) {$D_{m + 1}$};
  \node[usedBlock=7em,below=of D4] (U4) {$U_{m + 1}$};

  \node[gray,right=of D4] (D5) {$\cdots$};
  \node[gray,right=4.5em of U4] {$\cdots$};

  \draw[brace] ($(D0.north west) + (0,.6em)$) -- ($(D3.north east) + (0,.6em)$) node[braceCaptionAbove] {elements with ranks $0 \ldots 2^{m + 1} - 1$};
  \draw[brace] ($(U4.south east) + (0,-.3em)$) -- ($(U4.south west) + (0,-.3em)$) node[braceCaptionBelow] {rem.\ elements};
\end{tikzpicture}

%% file: section/offline-oram.tex

\section{Offline ORAM from Oblivious Priority Queues}%
\label{sec:OfflineOram}

As mentioned in the introduction, an offline ORAM is an oblivious array data structure given the sequence $I = \tuple{i_1, \ldots, i_n}$ of access locations in advance.
While an offline ORAM may pre-process $I$ to perform the operations more efficiently afterward, the data structure must still adhere to \cref{def:PerfectSecurity}, \ie the memory probes must be independent of the values in $I$.

An offline ORAM can be constructed from any oblivious priority queue using a technique similar to \emph{time-forward processing}~\cite{Chiang.Goodrich.ea95}.
We describe our construction for $N$ memory cells below:
In \cref{sec:OfflineOramOperations} we show how to realize \Call{Read}{$i_t$} and \Call{Write}{$i_t, v'_t$};
there we assume that the time $\tau_t$ at which the index $i_t$ is accessed next is known.
In \cref{sec:OfflineOramPreProcessing} we then show how to pre-process $I$ to derive these values $\tau_t$.

Jafargholi\ea~\cite{Jafargholi.Larsen.ea21} describe an alternative offline ORAM construction.
We simplify the construction by decoupling the information $\tau_t$ from the values written to the offline ORAM\@.

\subsection{Online Phase: Processing the Operations}%
\label{sec:OfflineOramOperations}

\begin{algorithm}[t]
  \caption{%
    Algorithm to perform an operation $\Call{Op}{} \in \set{\Call{Read}{}, \Call{Write}{}}$ in the offline ORAM at the access location $i$;
    $v'$ is the value to be written ($v' = \bot$ for $\Call{Op}{} = \Call{Read}{}$).
    The time-stamp $t$ is incremented after each access (with $t = 1$ initially).}%
  \label{alg:OfflineOramAccess}
  \begin{algorithmic}[1]
    \Procedure{Access}{$\Call{Op}{}, i, v'$}
      \State $\tuple{v, t_{\operatorname{next}}} \gets Q.\Call*{Min}{}$
      \State $Q.\Call*{DeleteMin}{}$ iff $t_{\operatorname{next}} = t$;
      perform a dummy operation iff $t_{\operatorname{next}} \not= t$
      \State $v \gets \begin{cases}
        v & \text{iff $\Call{Op}{} = \Call{Read}{} \land t_{\operatorname{next}} = t$,} \\
        v_{\operatorname{default}} & \text{iff $\Call{Op}{} = \Call{Read}{} \land t_{\operatorname{next}} \not= t$,} \\
        v' & \text{iff $\Call{Op}{} = \Call{Write}{}$}
      \end{cases}$
      \State $Q.\Call{Insert}{v, \Tau[t - 1]}$;%
        \label{alg:OfflineOramAccess:Insert}
      $t \gets t + 1$%
        \Comment{$\Tau[t - 1] = \tau_t$}
      \State\Return $v$
    \EndProcedure
  \end{algorithmic}
\end{algorithm}

For the offline ORAM with $N$ cells, initialize a priority queue $Q$ with capacity $N$;
the annotations $\Tau = \tuple{\tau_1, \ldots, \tau_n}$ as well as the current time $t$ are stored alongside the priority queue.
The procedure for processing the $t$-th operation \Call{Read}{$i_t$} or \Call{Write}{$i_t, v'_t$} is shown in \cref{alg:OfflineOramAccess}.
Some fixed value $v_{\operatorname{default}}$ is used as the initial value of all ORAM cells.

It is easy to verify that the resulting construction is correct given $\Tau$ and oblivious given a perfectly-secure priority queue $Q$.
For the array access $\Tau[t - 1]$ in \cref{alg:OfflineOramAccess:Insert}, note that $t$ is the number of the current operation, so the access can be performed \enquote{in the clear} without a linear scan;
this simplifies the construction w.\,r.\,t.\ Jafargholi\ea~\cite{Jafargholi.Larsen.ea21}.

\subsection{Offline Phase: Pre-Processing}%
\label{sec:OfflineOramPreProcessing}

To obtain the annotations $\Tau = \tuple{\tau_1, \ldots, \tau_n}$ we now describe how to pre-process the sequence $I = \tuple{i_1, \ldots, i_n} \in \set{0, \ldots, N - 1}^n$ of memory access locations.

The basic pre-processing proceeds as follows:
\begin{enumerate}
  \item Annotate each index $i_t$ with the time-stamp $t$.
  \item Obliviously sort the indices $I$ lexicographically by $i_t$ and $t$.
  \item Scan over indices in reverse, keeping track of the index $i$ and the time-stamp $t$, and annotate each index $i_t$ with the time-stamp $\tau_t$ it is accessed next (or some value larger than $n$ if there is no next access).
  \item Obliviously sort the indices $I$ by $t$ and discard everything but the annotations $\tau$.
\end{enumerate}
This can be done in amortized $\Oh(\log n)$ time per index with an $\Oh(n \log n)$-time oblivious sorting algorithm~\cite{Ajtai.Komlós.ea83} and results in the annotations $\Tau = \tuple{\tau_1, \ldots, \tau_n}$.

\begin{figure}[t]\centering
  \input{figure/blocked-offline-oram-preprocessing.tikz}
  \caption{%
    Pre-processing the sequence of memory access locations $I$ when $n$ is super-polynomial in $N$.
    In this figure, $\tuple{i, t}$ denotes a tuple of index $i$ and time-stamp $t$ while $\tau_j(i)$ denotes the time-stamp at which the index $i$ is accessed next in block $j$.}%
  \label{fig:BlockedOfflineOramPreprocessing}
\end{figure}
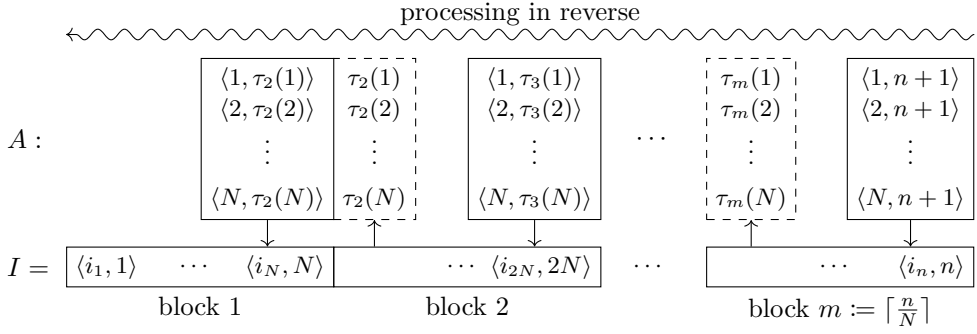

However, when the number of operations $n$ is super-polynomial in the capacity $N$, \ie when $n \in \om(N^c)$ for all constants $c$, the time per index exceeds the optimal runtime of $\Oh(\log N)$.
In this case, the pre-processing needs to be performed more carefully as shown in \cref{fig:BlockedOfflineOramPreprocessing} to maintain the amortized runtime of $\Oh(\log N)$:
We divide the sequence $I$ into blocks of size $N$.
Additionally, we maintain an auxiliary block $A$ with --- for each index $i \in \set{1, \ldots, N}$ --- the time-stamp $\tau$ at which index $i$ is accessed next.
Initially (for the last block), we initialize the time-stamp $\tau$ for each index $i$ to some value greater than $n$.

The time-stamps $\tau_t$ are then determined block by block, from the last to the first.
When processing each block, we update the time-stamps $\tau$ in $A$ for the block processed next.
For this, we can process the $\Oh(N)$ elements of each block (and $A$) as described above.

Since we process $\Oh(N)$ elements for each of the $\Oh(\ceil{\frac{n}{N}})$ blocks by sorting and scanning, the pre-processing has a runtime of $\Oh(n \log N)$ overall.
This maintains the desired runtime of $\Oh(\log N)$ per operation amortized.

With our priority queue construction from \cref{sec:ObliviousPriorityQueue}, we obtain the following:

\begin{theorem}[Optimal Offline ORAM]\label{thm:OptimalOfflineOram}
  There is a deterministic, perfectly-secure offline ORAM with capacity $N$ that has amortized $\Oh(\log N)$ overhead and uses $\Oh(N + n)$ space.
\end{theorem}

Note that in contrast to the optimal statistically-secure constructions~\cite{Jafargholi.Larsen.ea21,Shi20}, our offline ORAM maintains security and correctness for operation sequences of arbitrary length, \eg when $n$ is super-polynomial in $N$.

%% file: figure/blocked-offline-oram-preprocessing.tikz

\begin{tikzpicture}[
  block/.style={
    rectangle,draw,
    align=center,
    font=\small,
  },
]

  \matrix (elem) [
      matrix of nodes,
      nodes={
        block,
        minimum height=1.5em,
        inner sep=0,
      },
      column sep=-\pgflinewidth,
    ] {
    |[minimum width=10em]| $\cdots$ &
    |[minimum width=10em]| $\cdots$ &
    |[minimum width=4em,draw=none]| $\cdots$ &
    |[minimum width=10em]| $\cdots$ \\
  };
  \node[anchor=west,font=\small] at (elem-1-1.west) {$\tuple{i_1, 1}$};
  \node[anchor=east,font=\small] at (elem-1-1.east) {$\tuple{i_N, N}$};
  \node[anchor=east,font=\small] at (elem-1-2.east) {$\tuple{i_{2 N}, 2 N}$};
  \node[anchor=east,font=\small] at (elem-1-4.east) {$\tuple{i_n, n}$};

  \node[anchor=east] (I) at (elem-1-1.west) {$I = {}$};
  \node[anchor=north] at (elem-1-1.south) {block $1$};
  \node[anchor=north] at (elem-1-2.south) {block $2$};
  \node[anchor=north] at (elem-1-4.south) {block $m \coloneqq \ceil{\frac{n}{N}}$};

  \node[block,anchor=south east] (Amin) at ([yshift=1em] elem-1-4.north east) {%
    $\tuple{0, n + 1}$ \\
    $\tuple{1, n + 1}$ \\
    $\tuple{2, n + 1}$ \\
    \raisebox{1ex}{\vdots}};
  \draw[->] (Amin.south) -- (elem-1-4.north -| Amin.south);

  \node[block,anchor=south west,dashed] (Amout) at ([yshift=1em] elem-1-4.north west) {%
    $\tau_m(0)$ \\
    $\tau_m(1)$ \\
    $\tau_m(2)$ \\
    \raisebox{1ex}{\vdots}};
  \draw[->] (elem-1-4.north -| Amout.south) -- (Amout.south);
  \node[block,anchor=south east] (A2in) at ([yshift=1em] elem-1-2.north east) {%
    $\tuple{0, \tau_3(0)}$ \\
    $\tuple{1, \tau_3(1)}$ \\
    $\tuple{2, \tau_3(2)}$ \\
    \raisebox{1ex}{\vdots}};
  \draw[->] (A2in.south) -- (elem-1-2.north -| A2in.south);
  \node at ($(Amout.west)!.5!(A2in.east)$) {$\cdots$};

  \node[block,anchor=south west,draw=none] (A2out) at ([yshift=1em] elem-1-2.north west) {%
    $\tau_2(0)$ \\
    $\tau_2(1)$ \\
    $\tau_2(2)$ \\
    \raisebox{1ex}{\vdots}};
  \draw[dashed] ([yshift=.5\pgflinewidth] A2out.south west) -- ([yshift=.5\pgflinewidth] A2out.south east) -- ([yshift=-.5\pgflinewidth] A2out.north east) -- ([yshift=-.5\pgflinewidth] A2out.north west);
  \draw[->] (elem-1-2.north -| A2out.south) -- (A2out.south);
  \node[block,anchor=south east] (A1in) at ([yshift=1em] elem-1-1.north east) {%
    $\tuple{0, \tau_2(0)}$ \\
    $\tuple{1, \tau_2(1)}$ \\
    $\tuple{2, \tau_2(2)}$ \\
    \raisebox{1ex}{\vdots}};
  \draw[->] (A1in.south) -- (elem-1-1.north -| A1in.south);

  \node[anchor=west] at (A1in.west -| I.west) {$A : {}$};
  \draw[->,decorate,decoration={snake,amplitude=.2em}] ([yshift=8em] elem-1-4.north east) -- ([yshift=8em] elem-1-1.north west) node[midway,above] {processing in reverse};
\end{tikzpicture}

%% file: section/external-priority-queue.tex

\section{External-Memory Oblivious Priority Queue}%
\label{sec:ExternalObliviousPriorityQueue}

In many applications of oblivious algorithms and data structures, \eg for outsourced storage and for trusted computing in the presence of cache hierarchies, access to the main memory incurs high latencies.
In these applications, the complexity of an algorithm is more appropriately captured by the number of cache misses.
This motivates the study of oblivious algorithms in the \emph{external-memory}~\cite{Aggarwal.Vitter88} and \emph{cache-oblivious}~\cite{Frigo.Leiserson.ea99} models;
in this section we refer to these as \emph{cache-aware} and \emph{cache-agnostic} algorithms.

In this section, we instantiate I/O-efficient variants of our priority queue construction with perfect security and a private memory of constant size.
For this, we sketch how to obtain I/O-efficient partitioning algorithms with perfect security in \cref{sec:ExternalObliviousPartitioning}.
We then analyze the I/O-efficiency of our priority queue construction in \cref{sec:ExternalObliviousPriorityQueueAnalysis}.

\subparagraph{External-Memory Oblivious Algorithms}

In cache-aware and cache-agnostic models, the CPU operates on the data stored in an \emph{internal memory} (\emph{cache}) of $M$ memory words.
Blocks (\emph{cache-lines}) of $B$ memory words can be transferred between the internal and a large external memory (\emph{I/O operations}).
The number of these I/O operations, depending on the problem size $n$ as well as $M$ and $B$, is the primary performance metric for external algorithms~\cite{Aggarwal.Vitter88}.

Cache-aware algorithms depend on the parameters $M$ and $B$ and explicitly issue the I/O operations.
In contrast, cache-agnostic algorithms are unaware of the parameters $M$ and $B$;
here the internal memory is managed \enquote{automatically} through a \emph{replacement policy}~\cite{Frigo.Leiserson.ea99}.
We assume an optimal replacement policy and a \emph{tall cache}, \ie $M \geq B^{1 + \varepsilon}$ for a constant $\varepsilon > 0$;
both are standard assumptions~\cite{Frigo.Leiserson.ea99,Chan.Guo.ea18}.

For \emph{oblivious external-memory} algorithms, we apply the \cref{def:PerfectSecurity} to cache-aware and cache-agnostic algorithms.
In both cases we assume that the internal memory is conceptually distinct from the constant-size private memory.
That is, we guarantee that memory words both in the internal memory and within a block are accessed in an oblivious manner (\emph{strong obliviousness}~\cite{Chan.Guo.ea18}).
For this reason, our security definition remains unchanged.
Note that this implies that the block access pattern is also oblivious, \ie independent of the operations/inputs as in \cref{def:PerfectSecurity}.

\subparagraph{Previous Work}

\begin{table}[t]\centering%
  \begin{threeparttable}
    \begin{tabular*}{\textwidth}{@{\extracolsep\fill} c c c c c}
      \toprule
      & \multicolumn{2}{c}{\thead{cache-aware}} & \multicolumn{2}{c}{\thead{cache-agnostic}} \\
      \midrule
      \thead{statistical security} & $\Oh(\ceil*{\frac{n}{B}})$\tnote*{t, w} & \cite{Goodrich11} & $\Oh(\ceil*{\frac{n}{B}})$\tnote*{t, w} & \cite{Lin.Shi.ea19} \\
      \addlinespace[1ex]%
      \cmidrule{2-5}%
      \addlinespace[1ex]%
      \multirow{2}{*}{\thead{perfect security}} & $\Oh(\ceil*{\frac{n}{B}} \log_M n)$ & \cite{Lin.Shi.ea19} & $\Oh(\ceil*{\frac{n}{B}} \log_M n)$ & \cite{Lin.Shi.ea19} \\
      & $\Oh(\ceil*{\frac{n}{B}})$ & \textbf{new, via~\cite{Asharov.Komargodski.ea22}} & $\Oh(\ceil*{\frac{n}{B}} \log\log_M n)$\tnote*{t} & \textbf{new, via~\cite{Asharov.Komargodski.ea22}} \\
      \bottomrule
    \end{tabular*}
    \begin{tablenotes}[para]
      \item[t] assuming a \emph{tall cache} ($M \geq B^{1 + \varepsilon}$)
      \item[w] assuming a \emph{wide cache-line} ($B \geq \log^\varepsilon n$)
    \end{tablenotes}
    \caption{%
      Best known I/O upper bounds for obliviously partitioning $n$ elements.}%
  \label{tbl:ExternalObliviousPartitioningBounds}
  \end{threeparttable}
\end{table}

While there is a line of research explicitly considering cache-aware~\cite{Goodrich11,Goodrich.Simons14} and cache-agnostic~\cite{Chan.Guo.ea18,Lin.Shi.ea19} oblivious algorithms, most works on oblivious algorithms consider internal algorithms with runtime and bandwidth overhead as performance metrics.
To the best of our knowledge, cache-agnostic oblivious priority queues have not been explicitly considered in the literature.
Implicitly, I/O-efficiency is sometimes~\cite{Jacob.Larsen.ea19,Jafargholi.Larsen.ea21} treated through parameters:
An oblivious algorithm with $(B \mult w)$-width memory words and $M \mult w$ bits of private memory can equivalently be stated as an oblivious external-memory algorithm with $M$ words of internal memory and blocks of size $B$.
We note that this re-parameterization does not allow distinguishing internal and private memory and that the resulting algorithms are inherently cache-aware.

This equivalence allows us to restate upper and lower bounds in terms of external-memory algorithms:
Jacob\ea~\cite{Jacob.Larsen.ea19} show that $\Om(\frac{1}{B} \log \frac{N}{M})$ I/O operations amortized are necessary for a cache-aware oblivious priority queue;
this also applies to cache-agnostic oblivious priority queues.\footnote{%
  In contrast, $\Th(\frac{1}{B} \log_{\frac{M}{B}} \frac{N}{B})$ I/O operations amortized are sufficient for non-oblivious priority queues~\cite{Arge.Bender.ea02}.
  Here, the base of the logarithm is not constant but depends on the parameters $M$, $B$.}
This bound is matched by Jafargholi\ea~\cite{Jafargholi.Larsen.ea21}, but the construction is cache-aware, requires $\Om(\log N)$ words of private memory, and is randomized with statistical security.

The optimal internal partitioning algorithm~\cite{Asharov.Komargodski.ea22} has --- due to the use of expander graphs --- an oblivious, but highly irregular access pattern and is thus not I/O-efficient.
There are external oblivious partitioning algorithms~\cite{Goodrich11,Lin.Shi.ea19}, but they are either only statistically secure or inefficient.
We provide an overview of existing partitioning algorithms in \cref{tbl:ExternalObliviousPartitioningBounds}.

\subsection{External-Memory Oblivious Partitioning}%
\label{sec:ExternalObliviousPartitioning}

For I/O-efficient instantiations of our priority queue, we need I/O-efficient partitioning algorithms.
For this reason, we show how to construct an optimal cache-aware and a near-optimal cache-agnostic oblivious partitioning algorithm, respectively, with perfect security.
Remember that for partitioning with a predicate $P$, we need to permute the elements such that all elements $x$ with $P(x) = 0$ precede those with $P(x) = 1$.

\begin{algorithm}[t]
  \caption{%
    Cache-aware oblivious partitioning algorithm.
    For simplicity, we assume $m \geq 2$.}%
  \label{alg:CacheAwarePartition}
  \begin{algorithmic}[1]
    \Procedure{CacheAwarePartition$_P$}{$A$}%
      \State conceptually partition $A$ into $m \coloneqq \ceil{\frac{\abs{A}}{B}}$ blocks $G_i$ of $B$ consecutive elements each (where the last block may have fewer elements)
      \For{$i \gets 0, \ldots, m - 1$}
        \Call{Partition}[$P$]{$G_i$}
      \EndFor
      \For{$i \gets 1, \ldots, m - 2$}
        \Call{PurifyHalf}[$P$]{$G_{i - 1}, G_i$}%
          \label{alg:CacheAwarePartition:Consolidation}%
          \Comment{consolidate the blocks}
      \EndFor
      \State \Call{Partition}[$P'\colon X \mapsto P(X[0])$]{$\tuple{G_0, \ldots, G_{m - 3}}$}%
        \label{alg:CacheAwarePartition:BlockPartition}%
        \Comment{apply \Call{Partition}{} to the blocks}
      \State \Call{Reverse}{$G_{m - 1}$};
      \Call{PartBitonic}[$P$]{$G_{m - 2} \concat G_{m - 1}$}%
        \label{alg:CacheAwarePartition:MergeLast}
      \State \Call{Reverse}{$G_{m - 2} \concat G_{m - 1}$};
      \Call{PartBitonic}[$P$]{$G_0 \concat \cdots \concat G_{m - 1}$}
    \EndProcedure
  \end{algorithmic}
\end{algorithm}

We mainly rely on the optimal (internal) oblivious partitioning algorithm~\cite[Theorem~5.1]{Asharov.Komargodski.ea22} (\Call{Partition}{}) and standard external-memory techniques.
We also use oblivious building blocks from the previous work by Lin\ea~\cite[full version, Appendix~C.1.2]{Lin.Shi.ea19}:
\begin{description}
  \item[{\Call{PurifyHalf}[$P$]{$A, B$}}] This procedure is given two partitioned blocks $A$ and $B$ with $\abs{A} = \abs{B}$ and permutes the elements such that $A$ is \emph{pure}, \ie either only consists of elements $x$ with $P(x) = 0$ or only consists of elements with $P(x) = 1$, and $B$ is again partitioned.\footnote{%
    We note that, as confirmed by the authors (personal communication), Algorithm~5 given by Lin\ea\ for \Call{PurifyHalf}{} contains a technical error;
    the resulting block is not necessarily pure.
    This, however, can be fixed easily, see \cref{sec:CorrectPurifyHalf}.}

  \item[{\Call{PartBitonic}[$P$]{$A$}}] This procedure is given a \emph{bitonically partitioned}~\cite{Lin.Shi.ea19} array $A$, \ie an array where all elements $x$ with $P(x) = 1$ or all elements with $P(x) = 0$ are consecutive, and partitions $A$.
\end{description}
Both building blocks are deterministic with perfect security, cache-agnostic, and have a linear runtime of $\Oh(\ceil{\frac{n}{B}})$~\cite{Lin.Shi.ea19}.

Our cache-aware partitioning algorithm is shown in \cref{alg:CacheAwarePartition}.
The idea is to split $A$ into blocks of size $B$, partition each block, and then apply the internal partitioning algorithm to the blocks.

\begin{corollary}[Optimal Cache-Aware Oblivious Partitioning via~\cite{Asharov.Komargodski.ea22,Lin.Shi.ea19}]\label{cor:CacheAwareParition}
  There is a cache-aware, deterministic, perfectly-secure oblivious partitioning algorithm that requires $\Oh(\ceil{\frac{n}{B}})$ I/O operations for $n$ elements.
\end{corollary}

\begin{proof}
  For the correctness, note that after \cref{alg:CacheAwarePartition:Consolidation} all blocks except $G_{m - 2}$ and $G_{m - 1}$ are pure.
  By applying the internal partitioning algorithm to the blocks in \cref{alg:CacheAwarePartition:BlockPartition}, all $0$-blocks are swapped to the front.
  The partitioning is completed by first merging the partitions $G_{m - 2}$ and $G_{m - 1}$ in \cref{alg:CacheAwarePartition:MergeLast} and then merging both with the rest of the blocks in $A$.

  The partitioning of each individual block is performed in internal memory and thus requires $\Oh(\ceil{\frac{n}{B}})$ I/O operations overall.
  The consolidation and merging of the partitions can also be performed with $\Oh(\ceil{\frac{n}{B}})$ I/O operations~\cite{Lin.Shi.ea19}.
  For the partitioning in \cref{alg:CacheAwarePartition:BlockPartition}, the I/O-efficiency follows from the construction of the internal partitioning algorithm~\cite[Theorem~5.1]{Asharov.Komargodski.ea22}:
  The algorithm operates in a \enquote{balls-in-bins}-manner, \ie the elements are treated as indivisible.
  The algorithm performs a linear number of operations on $\floor{\frac{n}{B}}$ elements, where each element has size $\Oh(B)$.
  This leads to an I/O complexity of $\Oh(\ceil{\frac{n}{B}})$ overall.

  The obliviousness follows since the access pattern for each operation is a deterministic function of the input size $n \coloneqq \abs{A}$.
\end{proof}

\begin{algorithm}[t]
  \caption{%
    Cache-agnostic oblivious partitioning algorithm for $M \geq B^{1 + \varepsilon}$.
    We assume $m \geq 2$.}%
  \label{alg:CacheAgnosticPartition}
  \begin{algorithmic}[1]
    \Procedure{CacheAgnosticPartition$_P$}{$A$}
      \If{$n \leq 4$}
        \State compact $A$ via oblivious sorting
      \Else
        \State conceptually partition $A$ into $m \coloneqq \ceil{\frac{\abs{A}}{k}}$ groups $G_i$ of $k \coloneqq \ceil{\sqrt[1 + \varepsilon]{\abs{A}}}$ consecutive elements (where the last group may have fewer elements)
        \For{$i \gets 0, \ldots, m - 1$}
          \Call{CacheAgnosticPartition}[$P$]{$G_i$}
        \EndFor
        \For{$i \gets 1, \ldots, m - 2$}
          \Call{PurifyHalf}[$P$]{$G_{i - 1}, G_i$}%
            \label{alg:CacheAgnosticPartition:Consolidation}%
            \Comment{consolidate the groups}
        \EndFor
        \State \Call{Partition}[$P'\colon X \mapsto P(X[0])$]{$\tuple{G_0, \ldots, G_{m - 3}}$}%
          \Comment{apply \Call{Partition}{} to the groups}%
          \label{alg:CacheAgnosticPartition:GroupPartition}
        \State \Call{Reverse}{$G_{m - 1}$};
        \Call{PartBitonic}[$P$]{$G_{m - 2} \concat G_{m - 1}$}
        \State \Call{Reverse}{$G_{m - 2} \concat G_{m - 1}$};
        \Call{PartBitonic}[$P$]{$G_0 \concat \cdots \concat G_{m - 1}$}
      \EndIf
    \EndProcedure
  \end{algorithmic}
\end{algorithm}

For the cache-agnostic partitioning algorithm, the elements can be processed similarly.
Here the parameter $B$ is unknown, so the idea is to recursively divide into smaller groups until a group has size $\leq B$.
The resulting algorithm is shown in \cref{alg:CacheAgnosticPartition}.

\begin{corollary}[Cache-Agnostic Oblivious Partitioning via~\cite{Asharov.Komargodski.ea22,Lin.Shi.ea19}]\label{cor:CacheAgnosticPartition}
  Assuming a tall cache of size $M \geq B^{1 + \varepsilon}$ for a constant $\varepsilon > 0$, there is a cache-agnostic, deterministic, perfectly-secure oblivious partitioning algorithm that requires $\Oh(\ceil*{\frac{n}{B}} \log\log_M n)$ I/O operations for $n$ elements.
\end{corollary}

\begin{proof}
  The correctness of the base case is obvious.
  For the recursive case, the algorithm proceeds as the cache-aware \cref{alg:CacheAwarePartition} above, so the correctness can be seen as in \cref{cor:CacheAwareParition}.

  For the I/O complexity of \cref{alg:CacheAgnosticPartition:GroupPartition}, we distinguish two cases:
  \begin{description}
    \item[$k \leq B$] In this case $B \geq k \geq \sqrt[1 + \varepsilon]{n}$, so $n \leq B^{1 + \varepsilon} \leq M$ with the tall cache assumption.
    This means that the problem instance fits in the internal memory, and the step thus has an I/O complexity of $\Oh(\ceil*{\frac{n}{B}})$.
    \item[$k > B$] Here we can rely on the same insight as for the cache-aware partitioning, \ie that the internal algorithm performs $\Oh(\frac{n}{k})$ operations on elements of size $k \geq B$.
    This leads to an I/O complexity of $\Oh(\frac{n}{k} \mult \ceil*{\frac{k}{B}}) = \Oh(\frac{n}{B})$.
  \end{description}
  The other steps have an I/O complexity of $\Oh(\ceil*{\frac{n}{B}})$ as in the cache-aware algorithm.

  On depth $i$ of the recursion tree, the instance size is
  \begin{equation*}
    n_i \coloneqq \sqrt[{(1 + \varepsilon)}^i]{n}
    \quad\quad
    \text{so that on depth}
    \quad\quad
    i \geq \log_{1 + \varepsilon} \frac{\log n}{\log M}
    \in \Th(\log\log_M n)
  \end{equation*}
  the instances fit in the internal memory and no further I/O operations are required for the recursion.
  This leads to an I/O complexity of $\Oh(\ceil*{\frac{n}{B}} \log\log_M n)$ overall.

  As in \cref{cor:CacheAwareParition}, the obliviousness follows since the access pattern for each operation is a deterministic function of the input size $n \coloneqq \abs{A}$.
\end{proof}

\subsection{Analysis of the External-Memory Oblivious Priority Queue}%
\label{sec:ExternalObliviousPriorityQueueAnalysis}

As for the internal algorithm introduced in \cref{sec:ObliviousBuildingBlocks}, we can obtain efficient external oblivious algorithms for $k$-selection via partitioning~\cite[full version, Appendix~E.2]{Lin.Shi.ea19}.
We thus obtain cache-aware and cache-agnostic algorithms for $k$-selection with the same asymptotic complexities as the partitioning algorithms described above.

With these external algorithms, we can analyze the construction described in \cref{sec:ObliviousPriorityQueue} in the cache-aware and cache-agnostic settings:

\begin{theorem}[External-Memory Oblivious Priority Queues]\label{thm:ExternalObliviousPriorityQueue}
  There are deterministic, perfectly-secure oblivious priority queues with capacity $N$ that support each operation with I/O complexity $\Oh(\frac{1}{B} \log \frac{N}{M})$ amortized (cache-aware) or $\Oh(\frac{1}{B} \log \frac{N}{M} \log\log_M N)$ amortized (cache-agnostic), respectively.
\end{theorem}

\begin{proof}
  We prove the theorem via a slightly more general statement:
  Assuming the existence of a deterministic, perfectly-secure $k$-selection algorithm with I/O complexity $O_{\Call{KSelect}{}}(n) \in \Om(\frac{n}{B})$, there is a deterministic, perfectly-secure priority queue with capacity $N$ that supports each operation with I/O complexity $\Oh(\frac{O_{\Call{KSelect}{}}(3 N)}{N} \log \frac{N}{M})$ amortized.

  Since the external $k$-selection algorithms are functionally equivalent to the internal algorithm and oblivious, the correctness and obliviousness follows from \cref{thm:ObliviousPriorityQueue}.
  For the I/O complexity, note that the first $j \coloneqq \log_2 M - \Oh(1)$ levels of the priority queue fit into the $M$ cells of the internal memory.
  The operations \Call{Insert}{}, \Call{Min}{}, and \Call{DeleteMin}{} only operate on $D_0$ and $U_0$, so they do not require additional I/O operations.

  It remains to analyze the I/O complexity of rebuilding the data structure.
  For this, we only need to consider rebuilding the levels $m \geq j$ since rebuilding levels $m < j$ only operates on the internal memory.
  When rebuilding a level $m \geq j$, the levels $i < j$ can be stored to and afterward retrieved from the external memory with $\Oh(\frac{M}{B})$ I/O operations.
  Assuming optimal page replacement, the amortized number of I/O operations for rebuilding is bounded by
  \begin{align*}
    & \sum_{m = j}^{\ell - 1} \frac{O_{\Call{KSelect}{}}(2^{m + 1} + 2^m) + \sum_{i = 0}^{m - 1} O_{\Call{KSelect}{}}(2^{i + 2}) + c \mult \frac{2^m}{B}}{2^m}
    + \underbrace{\Oh\paren*{\frac{M}{B \mult 2^j}}}_{\text{levels $< j$}} \\
    \leq{} & \sum_{m = j}^{\ell - 1} \frac{\Oh(O_{\Call{KSelect}{}}(3 \mult 2^m))}{2^m}
    + \Oh\paren*{\frac{1}{B}}
    \in \Oh\paren*{(\ell - j) \mult \frac{O_{\Call{KSelect}{}}(3 N)}{N}} \\
    ={} & \Oh\paren*{\frac{O_{\Call{KSelect}{}}(3 N)}{N} \log \frac{N}{M}}
    \,\text{.}
  \end{align*}

  With the cache-aware $k$-selection via \cref{cor:CacheAwareParition} and the cache-agnostic $k$-selection via \cref{cor:CacheAgnosticPartition}, we obtain the claimed I/O complexities.
\end{proof}

With this, we obtain an optimal cache-aware oblivious priority queue and a near-optimal cache-agnostic oblivious priority queue, both deterministic and with perfect security.
Using a cache-agnostic, perfectly-secure oblivious sorting algorithm with (expected) I/O complexity $\Oh(\frac{n}{B} \log_{\frac{M}{B}} \frac{n}{B})$~\cite{Chan.Guo.ea18}\footnote{%
  Chan\ea~\cite{Chan.Guo.ea18} only describe a \emph{statistically-secure} sorting algorithm that works by randomly permuting and then sorting the elements.
  Since the failure can only occur when permuting, can be detected, and leaks nothing about the input, we can repeat the permuting step until it succeeds~\cite[Section~3.2.2]{Chan.Shi.ea21}.
  This leads to a perfectly-secure sorting algorithm with the same complexity in expectation.%
}, we can apply the same construction as in \cref{sec:OfflineOram} to obtain external offline ORAMs with the same (expected) I/O complexities as in \cref{thm:ExternalObliviousPriorityQueue}.
We exploit that the construction is a combination of sorting, linear scans, and time-forward processing.

%% file: section/conclusion.tex

\section{Conclusion and Future Work}%
\label{sec:Conclusion}

In this paper, we show how to construct an oblivious priority queue with perfect security and (amortized) logarithmic runtime.
While the construction is simple, it improves the state-of-the-art for perfectly-secure priority queues, achieving the optimal runtime.
The construction immediately implies an optimal offline ORAM with perfect security.
We extend our construction to the external-memory model, obtaining optimal cache-aware and near-optimal cache-agnostic I/O complexities.

\subparagraph{Future Work}

The optimal perfectly-secure partitioning algorithm~\cite{Asharov.Komargodski.ea22} has enormous constant runtime factors (in the order of $\gg 2^{111}$~\cite{Dittmer.Ostrovsky20}) due to the reliance on bipartite expander graphs.\footnote{%
  The algorithm of Dittmer and Ostrovsky~\cite{Dittmer.Ostrovsky20} has lower constant runtime factors, but is randomized and only achieves statistical security due to a (negligible) failure probability.}
Nevertheless, our construction can also be implemented efficiently in practice --- albeit at the cost of an $\Oh(\log N)$-factor in runtime --- by relying on merging~\cite{Batcher68} (instead of $k$-selection via linear-time partitioning).
We leave comparing such a practical variant to previous protocols as a future work.

On the theoretical side, a main open problem is to obtain a perfectly-secure oblivious priority queue supporting deletions (of arbitrary elements) in optimal $\Oh(\log N)$ time.
An additional open problem is the de-amortization of the runtime complexity.
Considering external oblivious algorithms, an open problem is to close the gap on cache-oblivious partitioning, \ie remove the remaining $\Oh(\log\log_M n)$ factor in I/O complexity for perfectly-secure algorithms.
We consider all of these interesting problems for future works.

\subparagraph{Acknowledgements}

We thank all anonymous reviewers for their constructive comments that helped to improve the presentation.

%% file: section/k-selection.tex

\section{Oblivious \texorpdfstring{\boldmath$k$}{k}-Selection}%
\label{sec:ObliviousKSelection}

We now present the linear-time oblivious $k$-selection algorithm and formally prove its correctness, linear runtime, and obliviousness.
The algorithm is based on Blum\ea's classical linear-time selection algorithm~\cite{Blum.Floyd.ea73}.
We use the approach sketched by Lin\ea~\cite{Lin.Shi.ea19} and instantiate it with the optimal, deterministic, perfectly-secure partitioning algorithm of Asharov\ea~\cite[Theorem~5.1]{Asharov.Komargodski.ea22}.
Let \Call{Partition}[$P$]{$A$} denote partitioning an array $A$ with a predicate $P$:
All elements $x \in A$ with $P(x) = 0$ are swapped to the front of the array, followed by all elements with $P(x) = 1$.

Instead of tackling the $k$-selection problem directly, we first consider an algorithm for selecting the single element $x_k$ with rank $k$ (\cref{alg:ObliviousRankK}).
Via \Call{Partition}[$x \mapsto x \geq x_k$]{$A$}, this algorithm can be used to realize the full $k$-selection.

For simplicity, we assume that the elements in $A$ are pairwise distinct and totally ordered;
for our application this is guaranteed as described in \cref{sec:NonDistinctPriorities}.
For general applications, a total order can easily be established by annotating each element with its index in $A$ and ordering the elements lexicographically by their value and their index.

\begin{algorithm}[H]
  \caption{Obliviously determine the element with rank $k$ in $A$.}%
  \label{alg:ObliviousRankK}
  \begin{algorithmic}[1]
    \Procedure{KElement}{$k, A$}%
      \Comment{$n \coloneqq \abs{A}$ and $0 \leq k < n$}
      \If{$n < 7$}
        \State sort $A$ obliviously
        \State\Return $A[k]$ (oblivious via a linear scan)
      \Else
        \State conceptually partition $A$ into $\ceil{\frac{n}{5}}$ groups $G_i$ of five consecutive elements each (where the last group may have fewer elements)
        \State $M \gets \text{array of length $\ceil{\frac{n}{5}}$}$
        \For{$i \gets 0, \ldots, \abs{M} - 1$}
          \State $M[i] \gets \text{median of $G_i$ (\eg via oblivious sorting)}$
        \EndFor
        \State $m \gets \Call{KElement}{\floor{\frac{\abs{M}}{2}}, M}$%
          \label{alg:ObliviousRankK:MedianOfMedians}%
          \Comment{median of medians}
        \State \Call{Partition}[$x \mapsto x \geq m$]{$A$}
        \State $c \gets \text{number of elements $x$ in $A$ with $x < m$}$
        \State \Call{Invert}{$A$} iff $k \geq c$, perform dummy swaps iff $k < c$
        \State $k \gets \begin{cases}
          k & \text{iff $k < c$,} \\
          k - \ceil{\frac{3 n - 20}{10}} & \text{iff $k \geq c$}
        \end{cases}$
        \State\Return \Call{KElement}{$k, A[0 \ldots \floor{\frac{7 n + 20}{10}} - 1]$}%
          \label{alg:ObliviousRankK:Recursion}
      \EndIf
    \EndProcedure
  \end{algorithmic}
\end{algorithm}

\begin{proof}[Proof of \cref{cor:ObliviousKSelection}]
  We now prove the desired properties of \cref{alg:ObliviousRankK}.

  \proofsubparagraph{Correctness}
  The correctness can be seen by induction over the input size $n$:
  For $n < 7$, \Call{KElement}{} trivially determines the element with rank $k$.
  For $n \geq 7$, we prove the correctness by bounding the sizes of the sets
  \begin{equation*}
    A_1 \coloneqq \setCond{x \in A}{x < m}
    \quad\text{and}\quad
    A_2 \coloneqq \setCond{x \in A}{x \geq m}
    \,\text{.}
  \end{equation*}

  Among the medians in $M$, the median of medians $m$ has rank
  \begin{equation*}
    r = \floor*{\frac{\ceil{\frac{n}{5}}}{2}}
  \end{equation*}
  since there are $\ceil{\frac{n}{5}}$ groups.
  Since all elements are pairwise distinct and totally ordered, $r$ medians are smaller than $m$ and $\ceil{\frac{n}{5}} - r = \ceil{\frac{n}{10}}$ medians are at least as big as $m$.

  \begin{figure}[t]\centering
    \input{figure/k-selection-groups.tikz}
    \caption{%
      Groups $G_i$ (sorted by their respective medians $M[i]$) and the median of medians $m$ in \cref{alg:ObliviousRankK}.
      The elements marked orange are smaller than $m$ and the elements marked blue are at least as big as $m$.\footnotemark}%
    \label{fig:KSelectionGroups}
  \end{figure}
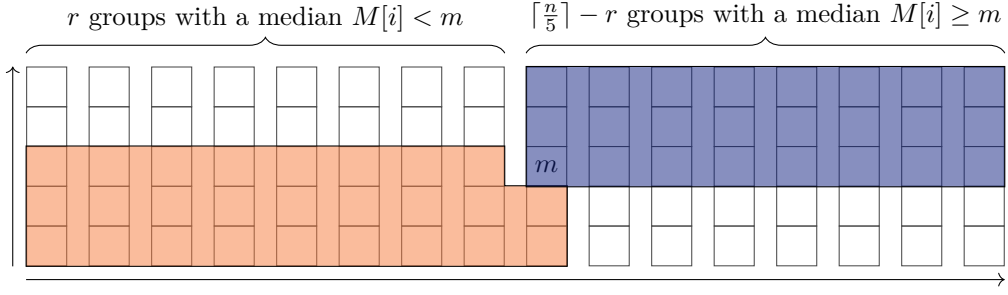
  \footnotetext{\Cref{fig:KSelectionGroups} is modeled after a similar figure for the original algorithm by Blum\ea~\cite[Figure~1]{Blum.Floyd.ea73}.}

  Except possibly for the last group, each median in $M$ is greater than two elements in its group as illustrated in \cref{fig:KSelectionGroups}.
  By transitivity, we can bound of size of $A_1$ as
  \begin{equation*}
    \abs{A_1} \geq 3 r + 2 - 2
    = 3 \floor*{\frac{\ceil{\frac{n}{5}}}{2}}
    \geq \frac{3 n - 15}{10}
    \,\text{.}
  \end{equation*}
  Similarly, we can bound the size of $A_2$ as
  \begin{equation*}
    \abs{A_2} \geq 3 \paren*{\ceil*{\frac{n}{5}} - r} - 2
    = 3 \ceil*{\frac{n}{10}} - 2
    \geq \frac{3 n - 20}{10}
    \,\text{.}
  \end{equation*}

  Depending on whether the element with rank $k$ is smaller or at least as large as the median of medians $m$, we can either exclude elements in $A_2$ or $A_1$.
  This means that always, at most $n - \frac{3 n - 20}{10} = \frac{7 n + 20}{10}$ elements remain.
  By inverting the partitioned array $A$ iff $k \geq c$, the $k$-th element is among the first $\floor{\frac{7 n + 20}{10}}$ elements in $A$ and the remaining $\ceil{\frac{3 n - 20}{10}}$ elements can be excluded.
  For $k \geq c$, the rank $k$ has to be adjusted by the number of excluded elements.

  Since $\ceil{\frac{n}{5}} \leq \floor{\frac{7 n + 20}{10}} < n$ for $n \geq 7$, the correctness of the recursion follows by induction.
  With $x_k \gets \Call{KElement}{k, A}$, $\Call{Partition}[x \mapsto x \geq x_k]{A}$ correctly swaps the $k$ smallest elements to the front.

  \proofsubparagraph{Runtime}
  In the recursive case $n \geq 7$, computing the medians $M$, partitioning~\cite[Theorem~5.1]{Asharov.Komargodski.ea22}, counting the elements $x < m$, and (conditionally) inverting $A$ can be done obliviously in linear time.
  With the two recursive calls (\cref{alg:ObliviousRankK:MedianOfMedians,alg:ObliviousRankK:Recursion}), for an overall linear runtime it thus suffices to show that
  \begin{equation*}
    \abs{M} + \floor*{\frac{7 n + 20}{10}}
    = \ceil*{\frac{n}{5}} + \floor*{\frac{7 n + 20}{10}} \leq c n
    \quad\text{for $n \geq n_0$ and some fixed $c < 1$.}
  \end{equation*}
  This holds, \eg\ for $c = 0.95$ and $n_0 = 50$.
  The runtime for $n < n_0$ is constant.

  \proofsubparagraph{Obliviousness}
  By using deterministic and oblivious sorting~\cite{Ajtai.Komlós.ea83} and partitioning~\cite{Asharov.Komargodski.ea22} algorithms, the obliviousness follows since the access pattern $\addr{\Call{KElement}{k, A}}$ is a deterministic function of the input size $n \coloneqq \abs{A}$.
\end{proof}

Although it is not necessary for our construction, the memory access pattern of \cref{alg:ObliviousRankK} is not only independent of the contents of $A$, but also of the parameter $k$.

%% file: figure/k-selection-groups.tikz

\begin{tikzpicture}[
  brace/.style={
    decorate,
    decorate,decoration={brace,amplitude=6pt},
  },
  braceCaption/.style={
    inner sep=0,
    align=center,
    midway,
    above,yshift=8pt,
  },
]
  \definecolor{xorange}{RGB}{245,121,58}
  \definecolor{xblue}{RGB}{15,32,128}

  \matrix (elem) [
      matrix of nodes,
      nodes in empty cells,
      nodes={
        rectangle,draw=black!70,
        anchor=north,
        minimum size=1.5em,
        inner sep=0,
      },
      row sep=-\pgflinewidth,
      column sep=.8em,
    ] {
    & & & & & & & & & & & & & & & \\
    & & & & & & & & & & & & & & & \\
    & & & & & & & & $m$ & & & & & & & \\
    & & & & & & & & & & & & & & & \\
    & & & & & & & & & & & & & & & \\
  };

  \draw[fill=xorange,fill opacity=.5] (elem-3-1.north west) -- (elem-3-8.north east) -- (elem-4-8.north east) -- (elem-4-9.north east) -- (elem-5-9.south east) -- (elem-5-1.south west) -- cycle;

  \draw[fill=xblue,fill opacity=.5] (elem-1-9.north west) -- (elem-1-16.north east) -- (elem-3-16.south east) -- (elem-3-9.south west) -- cycle;

  \draw[->] ([xshift=-.5em] elem-5-1.south west) -- ([xshift=-.5em] elem-1-1.north west);
  \draw[->] ([yshift=-.5em] elem-5-1.south west) -- ([yshift=-.5em] elem-5-16.south east);

  \draw[brace] ([yshift=.5em] elem-1-1.north west) -- ([yshift=.5em] elem-1-8.north east) node[braceCaption] {$r$ groups with a median $M[i] < m$};
  \draw[brace] ([yshift=.5em] elem-1-9.north west) -- ([yshift=.5em] elem-1-16.north east) node[braceCaption] {$\ceil{\frac{n}{5}} - r$ groups with a median $M[i] \geq m$};
\end{tikzpicture}

%% file: section/opq-operations.tex

\section{Oblivious Priority Queue Operations}%
\label{sec:ObliviousPriorityQueueOperations}

We now present pseudocode for the operations \Call{Insert}{}, \Call{Min}{}, and \Call{DeleteMin}{} of the oblivious priority queue.
We assume that the buffers $D_i$ and $U_i$, the number of levels $\ell$, and the counters $\Delta_i$ are attributes of the priority queue.
The elements are ordered by priority with dummy elements ordered after all non-dummy elements.
For brevity, we omit the additional time-stamp $t$ described in \cref{sec:NonDistinctPriorities}.

A new element is inserted by placing it into the empty up-buffer $U_0$ (\cref{alg:Insert}).

\begin{algorithm}[H]
  \caption{Insert the element $\tuple{k, p}$ into the oblivious priority queue.}%
  \label{alg:Insert}
  \begin{algorithmic}[1]
    \Procedure{Insert}{$k, p$}
      \State $U_0[0] \gets \tuple{k, p}$
      \For{$i \gets 0, \ldots, \ell - 1$}
        \State $\Delta_i \gets \Delta_i - 1$
      \EndFor
      \State\Call{Rebuild}{$\max\setCond{i < \ell}{\Delta_i = 0}$}
    \EndProcedure
  \end{algorithmic}
\end{algorithm}

Since the priority queue was rebuilt since last inserting or deleting an element, the minimal element can be found in the down-buffer $D_0$ (\cref{alg:Min}).
Note that --- for the correctness --- it is not necessary to reduce the counters $\Delta_i$ or rebuild the data structure after \Call*{Min}{}.

\begin{algorithm}[H]
  \caption{Determine the minimal element in the oblivious priority queue.}%
  \label{alg:Min}
  \begin{algorithmic}[1]
    \Procedure{Min}{{}}
      \State $\tuple{k, p_{\text{min}}} \gets \min D_0$
      \State\Return $\tuple{k, p_{\text{min}}}$
    \EndProcedure
  \end{algorithmic}
\end{algorithm}

Similarly, the minimal element can be removed by overwriting it with a dummy element $\bot$ (\cref{alg:DeleteMin}).
For obliviousness, the indexed access in \cref{alg:DeleteMin:IndexedAccess} is performed via a linear scan (over two elements).

\begin{algorithm}[H]
  \caption{Remove the minimal element from the oblivious priority queue.}%
  \label{alg:DeleteMin}
  \begin{algorithmic}[1]
    \Procedure{DeleteMin}{{}}
      \State $i_{\text{min}} \gets \arg\min D_0[\arb]$
      \State $D_0[i_{\text{min}}] \gets \bot$ (oblivious via a linear scan)%
        \label{alg:DeleteMin:IndexedAccess}
      \For{$i \gets 0, \ldots, \ell - 1$}
        \State $\Delta_i \gets \Delta_i - 1$
      \EndFor
      \State\Call{Rebuild}{$\max\setCond{i < \ell}{\Delta_i = 0}$}
    \EndProcedure
  \end{algorithmic}
\end{algorithm}

Note that both \Call{Insert}{} and \Call{DeleteMin}{} can be performed conditionally by \emph{conditionally} inserting the new element or \emph{conditionally} replacing the minimal element a dummy element.

To achieve operation-hiding security, a combined operation can (conditionally) perform all three operations \Call{Min}{}, \Call{Insert}{}, and \Call{DeleteMin}{} at once (without revealing the operations actually performed).
Here, only a single decrement of the counters $\Delta_i$ with one rebuilding step is necessary.

%% file: section/correct-purify-half.tex

\section{Correction of \texorpdfstring{\Call{PurifyHalf}{}}{PurifyHalf}}%
\label{sec:CorrectPurifyHalf}

\Cref{alg:PurifyHalf} below realizes \Call{PurifyHalf}{} as introduced by Lin\ea~\cite{Lin.Shi.ea19}:
The procedure expects two blocks $A$ and $B$ of equal size that are already partitioned.
After the invocation, the resulting block $A$ is pure, \ie either only consists of elements $x$ with $P(x) = 0$ or only consists of elements with $P(x) = 1$, and $B$ is partitioned.

\begin{algorithm}[H]
  \caption{%
    Corrected variant of \Call{PurifyHalf}{}~\cite[full version, Algorithm~5]{Lin.Shi.ea19}.
    Other than in the original algorithm we permute the elements in-place and use a predicate $P$ instead of markings.}%
  \label{alg:PurifyHalf}
  \begin{algorithmic}[1]
    \Procedure{PurifyHalf$_P$}{$A, B$}%
      \Comment{$n \coloneqq \abs{A} = \abs{B}$ and $P(x) \in \set{0, 1}$}
      \State $b \gets \text{majority of $P(x)$ over all elements $x \in A \concat B$}$
      \For{$i \gets 0, \ldots, n - 1$}
        \State obliviously swap $A[i]$ and $B[n - 1 - i]$ iff $P(A[i]) > P(B[n - 1 - i])$
      \EndFor
      \State obliviously swap $A$ and $B$ iff $b = 1$%
        \label{alg:PurifyHalf:ABSwap}
      \State \Call{PartBitonic}[$P$]{$B$}
    \EndProcedure
  \end{algorithmic}
\end{algorithm}

For the correctness, note that before \cref{alg:PurifyHalf:ABSwap} block $A$ is pure (with $P(x) = 0$) iff $b = 0$ and block $B$ is pure (with $P(x) = 1$) iff $b = 0$.
In both cases the respective other block is bitonically partitioned.
Since the procedure can be realized with linear scans and an invocation of \Call{PartBitonic}{}, the I/O complexity is in $\Oh(\ceil{\frac{n}{B}})$.

%% file: main.bbl
\begin{thebibliography}{10}

\bibitem{Aggarwal.Vitter88}
Alok Aggarwal and Jeffrey~Scott Vitter.
\newblock The input/output complexity of sorting and related problems.
\newblock {\em Communications of the {ACM}}, 31(9):1116--1127, 1988.
\newblock \href {https://doi.org/10.1145/48529.48535}
  {\path{doi:10.1145/48529.48535}}.

\bibitem{Ajtai.Komlós.ea83}
Miklós Ajtai, Jánoss Komlós, and Endre Szemerédi.
\newblock An {$O(n \log n)$} sorting network.
\newblock In David~S. Johnson, Ronald Fagin, Michael~L. Fredman, David Harel,
  Richard~M. Karp, Nancy~A. Lynch, Christos~H. Papadimitriou, Ronald~L. Rivest,
  Walter~L. Ruzzo, and Joel~I. Seiferas, editors, {\em Proceedings of the
  Fifteenth Annual {ACM} Symposium on Theory of Computing}, pages 1--9. {ACM},
  1983.
\newblock \href {https://doi.org/10.1145/800061.808726}
  {\path{doi:10.1145/800061.808726}}.

\bibitem{Arge.Bender.ea02}
Lars Arge, Michael~A. Bender, Erik~D. Demaine, Bryan Holland{-}Minkley, and
  J.~Ian Munro.
\newblock Cache-oblivious priority queue and graph algorithm applications.
\newblock In John~H. Reif, editor, {\em Proceedings of the Thirty-Fourth Annual
  {ACM} Symposium on Theory of Computing}, pages 268--276. {ACM}, 2002.
\newblock \href {https://doi.org/10.1145/509907.509950}
  {\path{doi:10.1145/509907.509950}}.

\bibitem{Asharov.Komargodski.ea22}
Gilad Asharov, Ilan Komargodski, Wei{-}Kai Lin, Kartik Nayak, Enoch Peserico,
  and Elaine Shi.
\newblock {OptORAMa}: Optimal oblivious {RAM}.
\newblock {\em Journal of the {ACM}}, 70(1):4:1--4:70, 2022.
\newblock \href {https://doi.org/10.1145/3566049} {\path{doi:10.1145/3566049}}.

\bibitem{Asharov.Komargodski.ea23}
Gilad Asharov, Ilan Komargodski, Wei{-}Kai Lin, and Elaine Shi.
\newblock Oblivious {RAM} with worst-case logarithmic overhead.
\newblock {\em Journal of Cryptology}, 36(2):7, 2023.
\newblock \href {https://doi.org/10.1007/s00145-023-09447-5}
  {\path{doi:10.1007/s00145-023-09447-5}}.

\bibitem{Batcher68}
Kenneth~E. Batcher.
\newblock Sorting networks and their applications.
\newblock In {\em Proceedings of the April 30--May 2, 1968, Spring Joint
  Computer Conference}, volume~32 of {\em {AFIPS} Conference Proceedings},
  pages 307--314. {ACM}, 1968.
\newblock \href {https://doi.org/10.1145/1468075.1468121}
  {\path{doi:10.1145/1468075.1468121}}.

\bibitem{Blum.Floyd.ea73}
Manuel Blum, Robert~W. Floyd, Vaughan~R. Pratt, Ronald~L. Rivest, and
  Robert~Endre Tarjan.
\newblock Time bounds for selection.
\newblock {\em Journal of Computer and System Sciences}, 7(4):448--461, 1973.
\newblock \href {https://doi.org/10.1016/S0022-0000(73)80033-9}
  {\path{doi:10.1016/S0022-0000(73)80033-9}}.

\bibitem{Boyle.Naor16}
Elette Boyle and Moni Naor.
\newblock Is there an oblivious {RAM} lower bound?
\newblock In Madhu Sudan, editor, {\em Proceedings of the 2016 {ACM} Conference
  on Innovations in Theoretical Computer Science}, pages 357--368. {ACM}, 2016.
\newblock \href {https://doi.org/10.1145/2840728.2840761}
  {\path{doi:10.1145/2840728.2840761}}.

\bibitem{Chan.Guo.ea18}
T.{-}H.~Hubert Chan, Yue Guo, Wei{-}Kai Lin, and Elaine Shi.
\newblock Cache-oblivious and data-oblivious sorting and applications.
\newblock In Artur Czumaj, editor, {\em Proceedings of the Twenty-Ninth Annual
  {ACM}-{SIAM} Symposium on Discrete Algorithms}, pages 2201--2220. {SIAM},
  2018.
\newblock \href {https://doi.org/10.1137/1.9781611975031.143}
  {\path{doi:10.1137/1.9781611975031.143}}.

\bibitem{Chan.Shi17}
T.{-}H.~Hubert Chan and Elaine Shi.
\newblock Circuit {OPRAM}: Unifying statistically and computationally secure
  {ORAMs} and {OPRAMs}.
\newblock In Yael Kalai and Leonid Reyzin, editors, {\em Theory of
  Cryptography}, volume 10678 of {\em Lecture Notes in Computer Science}, pages
  72--107. Springer, 2017.
\newblock \href {https://doi.org/10.1007/978-3-319-70503-3_3}
  {\path{doi:10.1007/978-3-319-70503-3_3}}.

\bibitem{Chan.Shi.ea21}
T.{-}H.~Hubert Chan, Elaine Shi, Wei{-}Kai Lin, and Kartik Nayak.
\newblock Perfectly oblivious (parallel) {RAM} revisited, and improved
  constructions.
\newblock In Stefano Tessaro, editor, {\em 2nd Conference on
  Information-Theoretic Cryptography ({ITC} 2021)}, volume 199 of {\em
  {LIPIcs}}, pages 8:1--8:23. Schloss Dagstuhl -- Leibniz-Zentrum für
  Informatik, 2021.
\newblock \href {https://doi.org/10.4230/LIPIcs.ITC.2021.8}
  {\path{doi:10.4230/LIPIcs.ITC.2021.8}}.

\bibitem{Chiang.Goodrich.ea95}
Yi{-}Jen Chiang, Michael~T. Goodrich, Edward~F. Grove, Roberto Tamassia,
  Darren~Erik Vengroff, and Jeffrey~Scott Vitter.
\newblock External-memory graph algorithms.
\newblock In Kenneth~L. Clarkson, editor, {\em Proceedings of the Sixth Annual
  {ACM}-{SIAM} Symposium on Discrete Algorithms}, pages 139--149. {SIAM}, 1995.

\bibitem{Dittmer.Ostrovsky20}
Samuel Dittmer and Rafail Ostrovsky.
\newblock Oblivious tight compaction in {$O(n)$} time with smaller constant.
\newblock In Clemente Galdi and Vladimir Kolesnikov, editors, {\em Security and
  Cryptography for Networks}, volume 12238 of {\em Lecture Notes in Computer
  Science}, pages 253--274. Springer, 2020.
\newblock \href {https://doi.org/10.1007/978-3-030-57990-6_13}
  {\path{doi:10.1007/978-3-030-57990-6_13}}.

\bibitem{Frigo.Leiserson.ea99}
Matteo Frigo, Charles~E. Leiserson, Harald Prokop, and Sridhar Ramachandran.
\newblock Cache-oblivious algorithms.
\newblock In {\em 40th Annual Symposium on Foundations of Computer Science},
  pages 285--297. {IEEE}, 1999.
\newblock \href {https://doi.org/10.1109/SFFCS.1999.814600}
  {\path{doi:10.1109/SFFCS.1999.814600}}.

\bibitem{Gale.Shapley62}
D.~Gale and L.~S. Shapley.
\newblock College admissions and the stability of marriage.
\newblock {\em The American Mathematical Monthly}, 69(1):9--15, 1962.
\newblock \href {https://doi.org/10.1080/00029890.1962.11989827}
  {\path{doi:10.1080/00029890.1962.11989827}}.

\bibitem{Goldreich.Ostrovsky96}
Oded Goldreich and Rafail Ostrovsky.
\newblock Software protection and simulation on oblivious {RAMs}.
\newblock {\em Journal of the {ACM}}, 43(3):431--473, 1996.
\newblock \href {https://doi.org/10.1145/233551.233553}
  {\path{doi:10.1145/233551.233553}}.

\bibitem{Goodrich11}
Michael~T. Goodrich.
\newblock Data-oblivious external-memory algorithms for the compaction,
  selection, and sorting of outsourced data.
\newblock In Rajmohan Rajaraman and Friedhelm~Meyer auf~der Heide, editors,
  {\em Proceedings of the Twenty-Third Annual {ACM} Symposium on Parallelism in
  Algorithms and Architectures}, pages 379--388. {ACM}, 2011.
\newblock \href {https://doi.org/10.1145/1989493.1989555}
  {\path{doi:10.1145/1989493.1989555}}.

\bibitem{Goodrich.Simons14}
Michael~T. Goodrich and Joseph~A. Simons.
\newblock Data-oblivious graph algorithms in outsourced external memory.
\newblock In Zhao Zhang, Lidong Wu, Wen Xu, and Ding{-}Zhu Du, editors, {\em
  Combinatorial Optimization and Applications}, volume 8881 of {\em Lecture
  Notes in Computer Science}, pages 241--257. Springer, 2014.
\newblock \href {https://doi.org/10.1007/978-3-319-12691-3_19}
  {\path{doi:10.1007/978-3-319-12691-3_19}}.

\bibitem{Ichikawa.Ogata23}
Atsunori Ichikawa and Wakaha Ogata.
\newblock Perfectly secure oblivious priority queue.
\newblock {\em {IEICE} Transactions on Fundamentals of Electronics,
  Communications and Computer Sciences}, E106.A(3):272--280, 2023.
\newblock \href {https://doi.org/10.1587/transfun.2022CIP0019}
  {\path{doi:10.1587/transfun.2022CIP0019}}.

\bibitem{Islam.Kuzu.ea12}
Mohammad~Saiful Islam, Mehmet Kuzu, and Murat Kantarcioglu.
\newblock Access pattern disclosure on searchable encryption: Ramification,
  attack and mitigation.
\newblock In {\em 19th Annual Network and Distributed System Security
  Symposium}. The Internet Society, 2012.
\newblock URL:
  \url{https://www.ndss-symposium.org/ndss2012/access-pattern-disclosure-searchable-encryption-ramification-attack-and-mitigation}.

\bibitem{Jacob.Larsen.ea19}
Riko Jacob, Kasper~Green Larsen, and Jesper~Buus Nielsen.
\newblock Lower bounds for oblivious data structures.
\newblock In Timothy~M. Chan, editor, {\em Proceedings of the 2019 Annual
  {ACM}-{SIAM} Symposium on Discrete Algorithms ({SODA})}, pages 2439--2447.
  {SIAM}, 2019.
\newblock \href {https://doi.org/10.1137/1.9781611975482.149}
  {\path{doi:10.1137/1.9781611975482.149}}.

\bibitem{Jafargholi.Larsen.ea21}
Zahra Jafargholi, Kasper~Green Larsen, and Mark Simkin.
\newblock Optimal oblivious priority queues.
\newblock In Dániel Marx, editor, {\em Proceedings of the 2021 {ACM}-{SIAM}
  Symposium on Discrete Algorithms ({SODA})}, pages 2366--2383. {SIAM}, 2021.
\newblock \href {https://doi.org/10.1137/1.9781611976465.141}
  {\path{doi:10.1137/1.9781611976465.141}}.

\bibitem{Keller.Scholl14}
Marcel Keller and Peter Scholl.
\newblock Efficient, oblivious data structures for {MPC}.
\newblock In Palash Sarkar and Tetsu Iwata, editors, {\em Advances in
  Cryptology --- {ASIACRYPT} 2014}, volume 8874 of {\em Lecture Notes in
  Computer Science}, pages 506--525. Springer, 2014.
\newblock \href {https://doi.org/10.1007/978-3-662-45608-8_27}
  {\path{doi:10.1007/978-3-662-45608-8_27}}.

\bibitem{Larsen.Nielsen18}
Kasper~Green Larsen and Jesper~Buus Nielsen.
\newblock Yes, there is an oblivious {RAM} lower bound!
\newblock In Hovav Shacham and Alexandra Boldyreva, editors, {\em Advances in
  Cryptology --- {CRYPTO} 2018}, volume 10992 of {\em Lecture Notes in Computer
  Science}, pages 523--542. Springer, 2018.
\newblock \href {https://doi.org/10.1007/978-3-319-96881-0_18}
  {\path{doi:10.1007/978-3-319-96881-0_18}}.

\bibitem{Lin.Shi.ea19}
Wei{-}Kai Lin, Elaine Shi, and Tiancheng Xie.
\newblock Can we overcome the {$n \log n$} barrier for oblivious sorting?
\newblock In Timothy~M. Chan, editor, {\em Proceedings of the 2019 Annual
  {ACM}-{SIAM} Symposium on Discrete Algorithms ({SODA})}, pages 2419--2438.
  {SIAM}, 2019.
\newblock \href {https://doi.org/10.1137/1.9781611975482.148}
  {\path{doi:10.1137/1.9781611975482.148}}.

\bibitem{Mitchell.Zimmerman14}
John~C. Mitchell and Joe Zimmerman.
\newblock Data-oblivious data structures.
\newblock In Ernst~W. Mayr and Natacha Portier, editors, {\em 31st
  International Symposium on Theoretical Aspects of Computer Science ({STACS}
  2014)}, volume~25 of {\em {LIPIcs}}, pages 554--565. Schloss Dagstuhl --
  Leibniz-Zentrum für Informatik, 2014.
\newblock \href {https://doi.org/10.4230/LIPIcs.STACS.2014.554}
  {\path{doi:10.4230/LIPIcs.STACS.2014.554}}.

\bibitem{Mondal.Panda.ea23}
Arup Mondal, Priyam Panda, Shivam Agarwal, Abdelrahaman Aly, and Debayan Gupta.
\newblock Fast and secure oblivious stable matching over arithmetic circuits.
\newblock Cryptology {ePrint} Archive, Paper 2023/1789, 2023.
\newblock URL: \url{https://eprint.iacr.org/2023/1789}.

\bibitem{Osvik.Shamir.ea06}
Dag~Arne Osvik, Adi Shamir, and Eran Tromer.
\newblock Cache attacks and countermeasures: The case of {AES}.
\newblock In David Pointcheval, editor, {\em Topics in Cryptology ---
  {CT}-{RSA} 2006}, volume 3860 of {\em Lecture Notes in Computer Science},
  pages 1--20. Springer, 2006.
\newblock \href {https://doi.org/10.1007/11605805_1}
  {\path{doi:10.1007/11605805_1}}.

\bibitem{Shi20}
Elaine Shi.
\newblock Path oblivious heap: Optimal and practical oblivious priority queue.
\newblock In {\em 2020 {IEEE} Symposium on Security and Privacy ({SP})}, pages
  842--858. {IEEE}, 2020.
\newblock \href {https://doi.org/10.1109/SP40000.2020.00037}
  {\path{doi:10.1109/SP40000.2020.00037}}.

\bibitem{Stefanov.Dijk.ea13}
Emil Stefanov, Marten van Dijk, Elaine Shi, Christopher~W. Fletcher, Ling Ren,
  Xiangyao Yu, and Srinivas Devadas.
\newblock Path {ORAM}: An extremely simple oblivious {RAM} protocol.
\newblock In Ahmad{-}Reza Sadeghi, Virgil~D. Gligor, and Moti Yung, editors,
  {\em Proceedings of the 2013 {ACM} {SIGSAC} Conference on Computer and
  Communications Security}, pages 299--310. {ACM}, 2013.
\newblock \href {https://doi.org/10.1145/2508859.2516660}
  {\path{doi:10.1145/2508859.2516660}}.

\bibitem{conferenceVersion}
Thore Thießen and Jan Vahrenhold.
\newblock Optimal offline {ORAM} with perfect security via simple oblivious
  priority queues.
\newblock In Julián Mestre and Anthony Wirth, editors, {\em 35th International
  Symposium on Algorithms and Computation ({ISAAC} 2024)}, volume 322 of {\em
  {LIPIcs}}, pages 55:1--55:18. Schloss Dagstuhl -- Leibniz-Zentrum für
  Informatik, 2024.
\newblock \href {https://doi.org/10.4230/LIPIcs.ISAAC.2024.55}
  {\path{doi:10.4230/LIPIcs.ISAAC.2024.55}}.

\bibitem{Toft11}
Tomas Toft.
\newblock Secure datastructures based on multiparty computation.
\newblock Cryptology {ePrint} Archive, Paper 2011/081, 2011.
\newblock URL: \url{https://eprint.iacr.org/2011/081}.

\bibitem{Wang.Nayak.ea14}
Xiao~Shaun Wang, Kartik Nayak, Chang Liu, T.{-}H.~Hubert Chan, Elaine Shi, Emil
  Stefanov, and Yan Huang.
\newblock Oblivious data structures.
\newblock In Gail{-}Joon Ahn, Moti Yung, and Ninghui Li, editors, {\em
  Proceedings of the 2014 {ACM} {SIGSAC} Conference on Computer and
  Communications Security}, pages 215--226. {ACM}, 2014.
\newblock \href {https://doi.org/10.1145/2660267.2660314}
  {\path{doi:10.1145/2660267.2660314}}.

\end{thebibliography}
